\documentclass[a4paper,UKenglish]{lipics-v2018}

\usepackage{microtype}
\usepackage[super]{nth}
\usepackage{mathrsfs}
\usepackage{amssymb}
\usepackage{hyperref}
\usepackage{thmtools}
\usepackage{thm-restate}
\usepackage[capitalize]{cleveref}
\usepackage{todonotes}

\definecolor{darkblue}{rgb}{0,0,0.5}
\hypersetup{colorlinks=true,citecolor=darkblue,linkcolor=darkblue,urlcolor=darkblue}

\DeclareMathOperator{\bw}{bw}
\DeclareMathOperator{\tw}{tw}

\newcommand{\bra}{\ensuremath{\langle}}
\newcommand{\ket}{\ensuremath{\rangle}}

\newcommand{\per}{\ensuremath{\operatorname{per}}}

\newcommand{\rk}{\ensuremath{\operatorname{rk}}}

\newcommand{\F}{\ensuremath{\mathbb{F}}}
\newcommand{\Z}{\ensuremath{\mathbb{Z}}}

\newcommand{\ts}{\ensuremath{\mathcal{T}_S}}
\newcommand{\td}{\ensuremath{\mathcal{T}_D}}
\newcommand{\tdk}{\ensuremath{\mathcal{T}_{D^{\otimes k}}}}

\declaretheorem[name=Theorem,numberwithin=section]{Thm}
\declaretheorem[name=Corollary,numberlike=Thm]{Cor}
\declaretheorem[name=Lemma,numberlike=Thm]{Lem}

\declaretheorem[name=Claim,numberlike=Thm]{Claim}

\newenvironment{Proof}{\begin{proof}}{\end{proof}}

\makeatletter
\newcommand{\slimfigure}[3]{%
	\@fleqnfalse
    \begin{equation}
    \begin{aligned}
      \label{fig:#3}
      \includegraphics[scale=#1]{#2}
    \end{aligned}
  \end{equation}
}
\makeatother

\title{Tensor network complexity of multilinear maps}

\titlerunning{Tensor network complexity of multilinear maps}

\author{Per Austrin}{School of Computer Science and Communication, KTH Royal Institute of Technology}{austrin@kth.se}{}{}

\author{Petteri Kaski}{Department of Computer Science, Aalto University}{petteri.kaski@aalto.fi}{}{}

\author{Kaie Kubjas}{Department of Mathematics and Systems Analysis, Aalto University\\Laboratoire d'Informatique de Paris 6, Sorbonne Universit\'e}{kaie.kubjas@aalto.fi}{}{}

\authorrunning{P. Austrin, P. Kaski, and K. Kubjas}

\Copyright{Per Austrin, Petteri Kaski, and Kaie Kubjas}

\subjclass{Theory of computation $\rightarrow$ Models of computation;
Theory of computation $\rightarrow$ Computational complexity and cryptography;
Theory of computation $\rightarrow$ Design and analysis of algorithms}

\keywords{arithmetic complexity, lower bound, multilinear map, tensor network}

\category{}

\relatedversion{An extended abstract of this paper appears in Proceedings of ITCS 2019.}

\supplement{}

\funding{%
Per Austrin was funded by the Swedish Research Council, Grant 621-2012-4546.  Petteri Kaski has received funding from the European Research Council, Grant 338077. Kaie Kubjas was supported by Marie Sk\l{}odowska-Curie  Grant 748354.}

\acknowledgements{%
We are grateful to Andreas Bj\"orklund for highlighting branchwidth to us as a natural parameter to generalize from clique-counting to counting homomorphisms.}

\nolinenumbers
\hideLIPIcs 

\begin{document}

\maketitle

%%%%%%%%%%%%%%%%%%%%%%%%%%%%%%%%%%%%%%%%%%%%%%%%%%%%%%%%%%%%%%%%%% Abstract %%%

\begin{abstract}
  We study \emph{tensor networks} as a model of arithmetic computation for
  evaluating multilinear maps.
  These capture any algorithm based on low border rank tensor decompositions, such as $O(n^{\omega+\epsilon})$ time matrix multiplication, and in addition many other algorithms such as $O(n \log n)$ time discrete Fourier transform and $O^*(2^n)$ time for computing the permanent of a matrix.
  However tensor networks sometimes yield faster algorithms than those
  that follow from low-rank decompositions.  For instance the
  fastest known $O(n^{(\omega +\epsilon)t})$ time algorithms for
  counting $3t$-cliques can be implemented with tensor networks, even
  though the underlying tensor has border rank $n^{3t}$ for all $t \ge 2$.
  For counting homomorphisms of a general pattern graph $P$ into a host graph on $n$ vertices we obtain an upper bound of $O(n^{(\omega+\epsilon)\bw(P)/2})$ where $\bw(P)$ is the branchwidth of $P$. This essentially matches the bound for counting cliques, and yields small improvements over previous algorithms for many choices of $P$.

  While powerful, the model still has limitations, and we are able to
  show a number of unconditional lower bounds for various multilinear maps,
  including:
  \begin{enumerate}
  \item[(a)] an $\Omega(n^{\bw(P)})$ time lower bound for counting homomorphisms from $P$ to an $n$-vertex graph, matching the upper bound if $\omega = 2$.
    In particular for $P$ a $v$-clique this yields an $\Omega(n^{\lceil 2v/3 \rceil})$ time lower bound for counting $v$-cliques, and for $P$ a $k$-uniform $v$-hyperclique we obtain an $\Omega(n^v)$ time lower bound for $k \ge 3$, ruling out tensor networks as an approach to obtaining non-trivial algorithms for hyperclique counting and the Max-$3$-CSP problem.

  \item[(b)] an $\Omega(2^{0.918n})$ time lower bound for the permanent
    of an $n \times n$ matrix.

  \end{enumerate}
\end{abstract}

%%%%%%%%%%%%%%%%%%%%%%%%%%%%%%%%%%%%%%%%%%%%%%%%%%%%%%%%%%%%% Document body %%%

\section{Introduction}

\noindent
One of the cornerstones of the theory of computation is the study of efficient
algorithms:
\begin{quote}
For a function $f$, how much time is required to evaluate $f$ on given inputs?
\end{quote}
Answering this question for almost any specific $f$ is well beyond
reach of contemporary tools. For example, it is theoretically possible
that canonical NP-complete problems,
such as the Circuit Satisfiability problem, can be solved in linear
time whereas they are widely believed to require super-polynomial
(or somewhat less widely, exponential) time \cite{Gasarch2012,Impagliazzo2001,Impagliazzo2001b}.
The main reason for this barrier to quantitative understanding is that it is
very hard to prove lower bounds for explicit functions in general models
of computation such as circuits or Turing machines.

This situation withstanding, a more modest program to advance our
understanding of computation is to study restricted models of computation
that for many $f$ of interest are simultaneously
\begin{enumerate}[(i)]
\item
general enough to capture the fastest-known algorithms for
$f$, and
\item
restricted enough to admit proofs of strong {\em unconditional} time
lower bounds for $f$.
\end{enumerate}
There is a substantial body of existing work that fits under this program,
ranging from the study of low-depth or otherwise restricted circuits (see e.g.~\cite{AroraB2009}, Ch.~14)
to models of algorithm-design principles such as greedy algorithms,
backtracking, or dynamic programming
\cite{AlekhnovichBBIMP2011,DavisI2009}, to linear or semidefinite
programming relaxations for hard combinatorial optimization problems
\cite{LeeRS2015}.

\subsection{Multilinear maps}

One class of functions $f$ that are of substantial importance is the family of {\em $\ell$-linear maps} ({\em multilinear maps}) from $\ell$ input vector spaces to an output vector space.%
\footnote{Multilinear maps with $\ell=1$ are called {\em linear}, $\ell=2$ {\em bilinear}, $\ell=3$ {\em trilinear}, and so forth.}{}
Examples range from maps of known near-linear-time complexity in the input size, such as the discrete Fourier transform~\cite{CooleyT1965,VanLoan1992}, to maps without known polynomial-time-complexity algorithms, such as the permanent of a matrix~\cite{Ryser1963,Valiant1979}. Beyond motivation in numerical multilinear algebra and its applications, recent advances in the study of fine-grained algorithm design and complexity have highlighted the fundamental role of algebraic methods in the fastest-known algorithm designs across a broad range of tasks from graph problems, such as all-pairs shortest paths and $k$-clique, to parsing and constraint satisfaction problems, such as maximum satisfiability and graph coloring~\cite{AbboudBW15a,BjorklundHKK2008,BjorklundHK2009,EisenbrandG2004,ItaiR1978,Nesetril1985,Williams2005,Williams14apsp}.

In this paper, we study the \emph{arithmetic complexity} of evaluating
a multilinear map, that is, the number of operations (scalar additions,
subtractions, and multiplications) needed to evaluate the map.
To set up a baseline, a generic $\ell$-linear map from $\ell$ vector spaces
of dimension $n$ to a scalar requires $\Omega(n^\ell)$ scalars to represent
the map directly using combinations of basis vectors. Given this
complexity of a direct explicit representation, it is a fundamental
problem to seek less costly representations, along with associated
efficient algorithms that work on the chosen representation.

We propose the systematic study of {\em tensor networks} on hypergraphs as
a framework for fast evaluation of multilinear maps, and show a number of
upper and lower bounds on the computational power of tensor networks
in the spirit of (i) and (ii) above.

\subsection{Tensor networks}

\label{sect:nets}

Tensor networks have a long and rich history which can be traced as
far back as \nth{19}-century studies in invariant theory due
to Cayley~\cite{Cayley1875,Cayley1879}, Clebsch~\cite{Clebsch1861},
Clifford~\cite{Clifford1878}, Sylvester~\cite{Sylvester1878},
and Kempe~\cite{Kempe1885,Kempe1892}.
Tensor networks are extensively deployed in applications from pure mathematics
and theoretical physics~\cite{Kauffman1989,Kuperberg1991,Landsberg2012,LandsbergQY2012,Penrose1956,Penrose1971,RichterGebertL2009,Schrijver2015a} to computational physics and chemistry~\cite{Orus2014,PfeiferHV2014,SolomonikMHSD2014}.
In theoretical computer science,
they appear in various guises including, for example,
in the Holant framework \cite{Valiant2008,CaiLX2011,CaiGW2016a},
in the study of probabilistic graphical models~\cite{KollerF2009,RobevaS2017},
in the study of parallel programming~\cite{SolomonikH2015},
in the study of quantum computing~\cite{ItaiL2010},
and in the study of arithmetic
complexity~\cite{BeaudryH2007,CapelliDM2016,DammHM2002}.
Tensor contraction is also emerging as a basic
computational primitive in computer
hardware~\cite{JouppiYPPABBBBB17,MarkidisCLPV18}.
(See \cref{sec:related work} for a detailed discussion.)
As the precise definitions are somewhat technical, let us start with
a few simple motivating examples and then state our results,
with the understanding that precise definitions appear
in \cref{sec:tensor network defs}.

In our setting, a tensor network is a hypergraph in which the vertices are
tensors and the hyperedges are called \emph{modes}.
Each mode that is incident to a tensor defines a ``dimension''
for indexing the entries of the tensor---for example, a matrix is a tensor
that is incident to two modes, one mode for the rows of the matrix, and
the other mode for the columns of the matrix. A network may be simplified
by a sequence of {\em contractions}, where each contraction takes a subset
of tensors and replaces it with a single tensor whose entries are obtained
as generalized inner products of the entries of the tensors being contracted.

As a concrete first example of these concepts, let us consider the task of
multiplying two matrices, $A$ and $B$. More specifically, let
$A$ be a matrix with rows indexed by mode $i$ and
columns indexed by mode $k$, and let $B$ be a matrix
with rows indexed by mode $k$ and columns indexed by mode $j$.
We may represent the multiplication task as
the tensor network depicted on the left in \eqref{fig:intro matrix mul 1}.
The result of contracting $A$ and $B$ is a new
matrix with rows indexed by $i$ and columns indexed by $j$, where the
entry at each position $(i,j)$ is $\sum_k A_{ik} B_{kj}$.  If the three
index sets all have size $n$, then computing $A \cdot B$ by contracting
them in such a direct manner uses $\Theta(n^3)$ operations.  To obtain
faster matrix multiplication, we can rewrite the bare-bones network
on the left in \eqref{fig:intro matrix mul 1} using a low-rank
decomposition of the matrix multiplication tensor.
For example, Strassen's decomposition~\cite{Strassen1969}
of $2 \times 2$ matrix multiplication
can be represented using the second network in \eqref{fig:intro matrix mul 1}.
Note that the index $i$ used by $A$ and the result has been separated into
two distinct indices $i$ and $i'$, and similarly for $j$ and $k$.
\slimfigure{0.9}{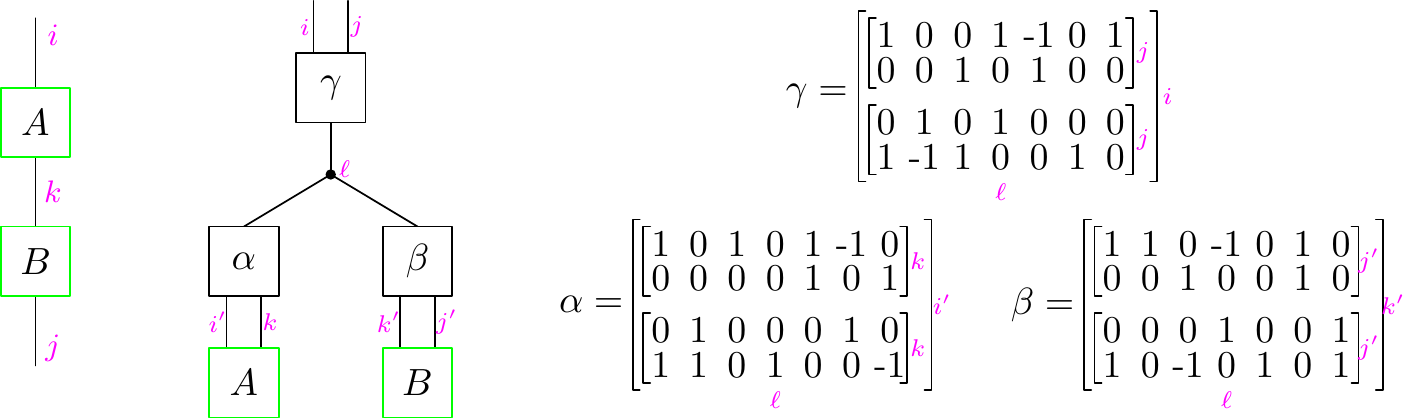}{intro matrix mul 1}

We can \emph{execute} the network by succesively contracting
groups of vertices.  In \eqref{fig:intro matrix mul 2} we see the
process of successively contracting pairs of tensors in a carefully
chosen order, until only a single tensor -- the result of the
computation -- remains.  Such an execution can be naturally
represented by a rooted binary tree, as shown on the right in
\eqref{fig:intro matrix mul 2}, where the tensors of the network form
the leaves, and each internal node represents the result of
contracting its two children.  To summarize, a tensor-network
algorithm is specified by providing (a) a tensor network that when
contracted yields the desired result, \emph{and} (b) an execution tree
indicating the order in which to contract tensors in the network.

\slimfigure{0.74}{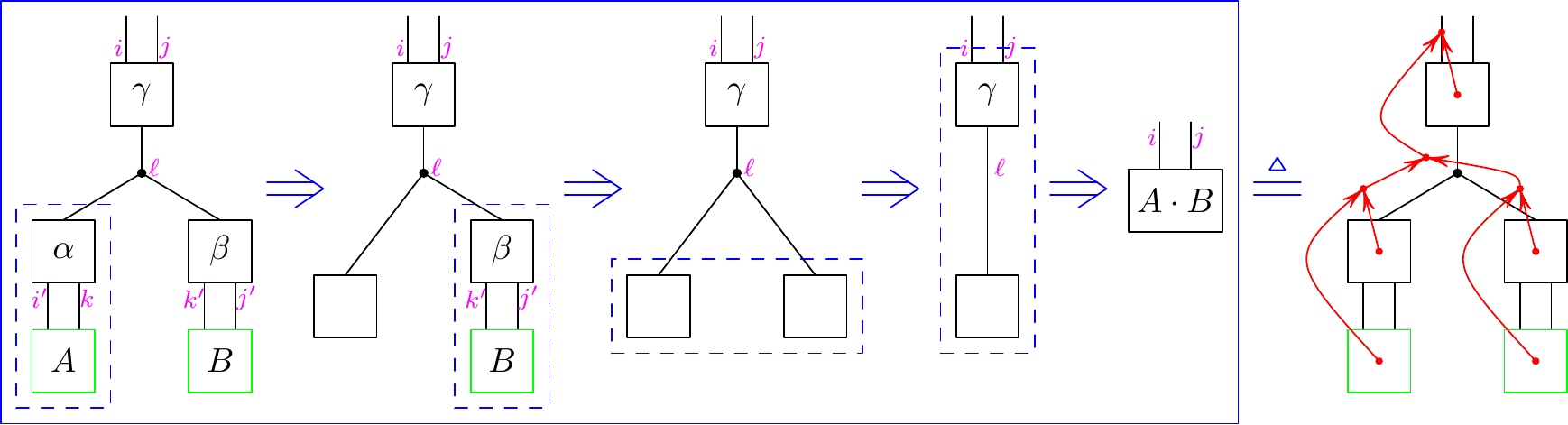}{intro matrix mul 2}

The \emph{cost} of performing one of the contractions in an execution
is the product of the lengths of the modes used by any tensor involved
in the contraction.  This simply measures (up to a constant factor)
the number of arithmetic operations (additions/multiplications) used
to compute the result by a direct, na\"\i{}ve computation that does not
depend on the values of the tensors.  For example, the contraction of
$\alpha$ and $A$ in the first step of \eqref{fig:intro matrix mul 2}
has cost $28$ because it involves the three modes $i'$ (length $2$),
$k$ (length $2$) and $\ell$ (length $7$).

We observe that cost is data-oblivious -- the tensor $\alpha$ is fixed with many
zero-entries but these entries still contribute to the cost.  Indeed, in many cases there may be faster ways of evaluating a contraction than to evaluate it naively, and just like we saw above, this can often be dealt with by rewriting the network appropriately.
For instance, consider now the multiplication of two $2^k \times 2^k$
matrices.  Because the family of matrix multiplication tensors is closed under
Kronecker products, this operation may be computed by a tensor
network like the one shown in \eqref{fig:intro matrix mul 3}
(depicting the case $k=5$), where $\alpha$, $\beta$ and $\gamma$ are
as in \eqref{fig:intro matrix mul 2}.  The rows/columns of the
matrices are now indexed by $k$-tuples of bits.  The execution of this
network contracts one $\alpha$/$\beta$/$\gamma$ tensor at a time,
which lets us keep the cost low.  For example, the first contraction of $A$
with the first $\alpha$ block has a cost of $2^k \cdot 2^k \cdot 7$,
and results in a tensor of size $2^{k-1} \times 2^{k-1} \times 7$,
then the next contraction has a cost of $2^{k-1} \cdot 2^{k-1} \cdot
7^2$ and produces a result of size $2^{k-2} \times 2^{k-2} \times 7
\times 7$, and so on, until the contraction with the last $\alpha$
block which has a cost of $2 \cdot 2 \cdot 7^k = O(7^k)$, and all the
contractions in the execution have cost bounded by this, meaning that
we get a total running time of $O(k 7^k) = O(N^{\log_2 7} \log N)$ for
$N \times N$ matrices.%
\footnote{In fact, a more careful analysis gives running time $O(N^{\log_2 7})$.}{}
\slimfigure{0.65}{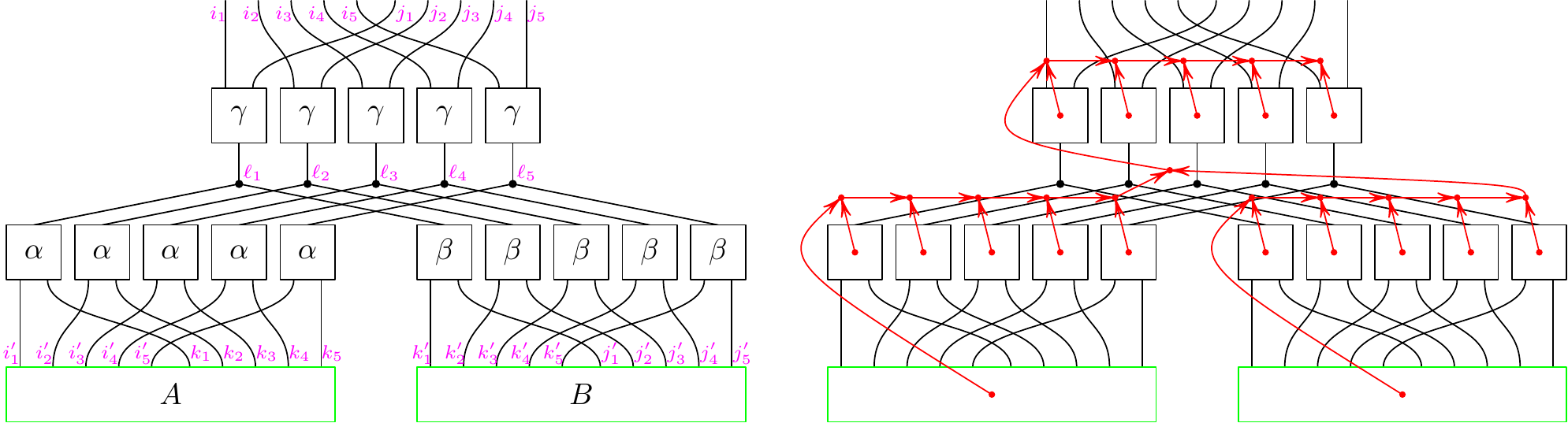}{intro matrix mul 3}

This type of argument can capture any algorithm based on a low-rank
decomposition of the underlying tensor of the multilinear map, and
indeed, this enables $O(n^\omega)$-time\footnote{Throughout the paper,
$\omega = \omega(\F)$ denotes the infimum over all $e$ such that
the arithmetic complexity of multiplying two $n \times n$ matrices is
$O(n^e)$.  While the value of $\omega$ may depend on the underlying
field $\F$, we tacitly ignore this, since the field is fixed throughout
the paper. For all fields it is known that
$2\leq\omega<2.3728639$~\cite{LeGall2014,VassilevskaWilliams2012}.}{}
 matrix multiplication
using tensor networks.
Beyond simple low-rank decompositions, which always give rise to
``star-like'' networks as in \eqref{fig:intro matrix mul 3}, there are
many interesting algorithms that can be captured using networks with a
more complicated topology.
For instance, many maps of substantial
importance have a layered structure that decomposes the map to a sequence
of elementary maps. A canonical example is the
discrete Fourier transform (DFT),
which for a smooth composite order such as $2^k$, can be decomposed
into a fast Fourier transform (FFT) that consists of a sequence
of $k$ transforms of order~$2$ interleaved with diagonal-matrix
multiplications of twiddle factors~\cite{CooleyT1965,VanLoan1992}.

\subsection{Our results}
\label{sect:our}

Starting with motivation (i) and seeking to express existing fastest-known
algorithms
as executions of tensor networks by a sequence of contractions,
we show upper bounds for a number of
natural problems. Beyond standard linear settings such as the FFT,
not only do tensor networks realize classical bilinear settings
such as Abelian group algebra products and
fast matrix multiplication algorithms based on low tensor rank,
they are powerful enough to capture a substantial number of higher-linearity
applications, including Ryser's algorithm for matrix permanent \cite{Ryser1963},
and the {\em Kruskal operator}~\cite{Kolda2006,Kruskal1977}
(cf.~\cref{sect:kruskal}),
which underlies realization of rank-decompositions for tensor rank
\cite{Kolda2009} and current fastest algorithms for detecting outlier
correlations~\cite{KarppaKK2018}.

One problem for which tensor networks turn out to be particularly useful
is counting homomorphisms from a fixed pattern graph $P$ to
a large host graph $G$ on $n$ vertices.  The most well-studied such problem
is when $P$ is a $k$-clique.  For this problem, the currently fastest
algorithm runs in time roughly $O(n^{\omega k/3})$ (with variations in
the exponent depending on $k \bmod 3$)
\cite{Nesetril1985,EisenbrandG2004}.  For general $P$, it is known
that the problem can be solved in $O(n^{\tw(P)+1})$ time
\cite{DiazST02}, where $\tw(P)$ is the treewidth of $P$.
We show that tensor networks can solve the problem in $O(n^{(\omega+\epsilon)
  \bw(P)/2})$ time, where $\bw(P)$ is the \emph{branchwidth} of $P$.
For $P$ a $k$-clique we have $\bw(P) = \lceil 2k/3\rceil$ so this
almost recovers the $O(n^{\omega k/3})$ running time, and in this case
we can slightly improve the running time to recover the
$O(n^{\omega \lfloor k/3 \rfloor + (k \bmod 3)})$ time of Nešetřil and
Poljak \cite{Nesetril1985}.  In the case of general $P$, this improves
on the treewidth-based bound for graphs with $\bw(P) \le
2(\tw(P)+1)/\omega$ (and in particular if $\omega=2$ it is always as
fast as the treewidth-based bound, ignoring the $\epsilon$).  By
recent results of Curticapean, Dell, and Marx~\cite{CurticapeanDM2017},
fast algorithms for homomorphism-counting can be used to obtain fast
algorithms for counting subgraphs of $G$ isomorphic to $P$, and in
some cases our new branchwidth-based bound leads to an improvement;
for example, for counting paths of lengths of length $7$, $8$ or $9$,
we get a running time of $O(n^{3\omega/2+\epsilon}) < O(n^{3.56})$
compared to $O(n^{4})$ using the treewidth-based bound, whereas for very long
paths it is not clear whether we would need $\omega=2$ in order for
these bound to improve on the treewidth-based bound. Previous work
that combines branch decompositions and fast matrix
multiplication includes Dorn~\cite{Dorn2006} and
Bodlaender~{\em et~al.}~\cite{BodlaenderLRV2010}.

Further applications
captured by tensor networks are the set covering and set partitioning
frameworks via fast zeta and M\"obius transforms that underlie the
current fastest algorithms for graph coloring~\cite{BjorklundHK2009}
and its generalizations such as computing the
Tutte polynomial~\cite{BjorklundHKK2007,BjorklundHKK2008}.
To summarize, we have the following compendium theorem of upper bound results.
For the detailed definitions of the relevant multilinear maps,
cf.~\cref{sect:examples} and \cref{sec:upper bounds}.

\begin{Thm}
  \label{thm:upper bounds}
  We have the following upper bounds on arithmetic complexity via tensor networks:
  \begin{enumerate}
  \item $O(n^{\omega + \epsilon})$ for the matrix multiplication map of two $n \times n$ matrices.
  \item $O(n^{(\omega+\epsilon)\lfloor v/3 \rfloor  + (v \bmod 3)})$ for counting $v$-cliques in an $n$-vertex graph.
  \item $O(n^{(\omega+\epsilon) \bw(P)/2})$ for counting homomorphisms of a fixed pattern (hyper)graph $P$ into a (hyper)graph on $n$ vertices.
  \item $O(\max(n^{\lceil\ell/2\rceil(\omega+\epsilon-1)}r,\,n^{2\lceil\ell/2\rceil}r^{\omega+\epsilon-2}))$ for the Kruskal operator of $\ell$ matrices of shape $n\times r$.
  \item $O(2^k k)$ for the discrete Fourier transforms for the Abelian groups $\Z_{2^k}$ and $\Z_2^k$.
  \item $O(2^kk)$ for group algebra products on $\F[\Z_{2^k}]$ and $\F[\Z_2^k]$ when $2$ is unit in $\F$.
  \item $O(2^kk)$ for the semigroup algebra product on $\F[(\{0,1\}^k,\subseteq,\cap,\cup)]$.
  \item $O(2^n n)$ for the permanent of an $n \times n$ matrix.

  \end{enumerate}
Above $\epsilon>0$ is an arbitrary constant.
\end{Thm}

Perhaps the most interesting application above
is the $v$-clique problem, which suggests that one should seek to
pursue generalizations to $v$-vertex hypercliques of
$\binom{v}{k}$ hyperedges with $v>k\geq 3$.
Indeed, subgraph counting is a problem that has received substantial
interest over the years (e.g.~\cite{%
ItaiR1978,%
Nesetril1985,%
AlonYZ1997,%
AlonG2010,%
EisenbrandG2004,%
BjorklundHKK2009,%
BjorklundKK2014,%
WilliamsWWY15,%
WilliamsW2013,%
FominLRSR2012,%
FloderusKLL2015,%
Olariu1988,%
KloksKM2000,%
CurticapeanDM2017}), but progress in the particular case of $v$-clique
has been stuck to the extent that the problem has attracted notoriety as
a hardness assumption in fine-grained
complexity~\cite{AbboudBBK2017,AbboudBW15a}.
Beyond the study of cliques, hypercliques, and subgraph counting, nontrivial
algorithms for such forms would have immediate applicability, for example,
in the study of maximum constraint satisfaction problems (Max-CSP) for
constraints of width $k\geq 3$; cf.~Williams~\cite{Williams2005} for
the case $k=2$. One of the main goals of our subsequent lower bounds is
to rule out tensor networks as a candidate to yield improved algorithms
in this setting.

Turning from upper bounds to lower bounds and motivation (ii), tensor
networks are restricted enough to enable nontrivial lower bounds for many
multilinear maps.
To begin with, an immediate limitation of tensor networks is
that all the intermediate results during the execution of a network are
multilinear, and the execution of a network can be simulated by a multilinear
circuit. Raz~\cite{Raz2013} shows that multilinear formulas cannot compute the
determinant of an $n\times n$ matrix in a polynomial number of operations
in $n$, even though polynomial-size general circuits are known for
the determinant (cf.~\cite{Berkowitz1984,BunchH1974,Rote2001}).

It turns out that considerably stronger lower bounds can be shown
for tensor networks. In particular, we establish essentially tight lower bounds
(subject to the assumption $\omega=2$) for arithmetic complexity via tensor
networks of $P$-homomorphism counting and the Kruskal operator.
Furthermore, we rule out the possibility of any
nontrivial algorithm designs via tensor networks for counting cliques in hypergraphs.
The following theorem collects
our main lower-bound results, and should be contrasted with the upper bounds
in Theorem~\ref{thm:upper bounds}.

\begin{Thm}
  \label{thm:lower bounds}
  We have the following lower bounds on arithmetic complexity via tensor networks:
  \begin{enumerate}
  \item $\Omega(n^{\bw(P)})$ for the multilinear form corresponding to $P$-homomorphism counting.  In particular, this yields a lower bound of $\Omega(n^{\lceil 2v/3 \rceil})$ for counting cliques of size $v$, and a lower bound of $\Omega(n^v)$ for counting hypercliques of size $v$.
  \item $\Omega(\max(n^\ell, n^{\lceil l/2\rceil}r))$ for the Kruskal operator of $\ell$ matrices of shape $n\times r$.
  \item $\Omega(\binom{n}{n/3})$ for the determinant or permanent of an $n \times n$ matrix.
  \end{enumerate}
\end{Thm}

We remark that \cite{LincolnWW18} independently showed that the border rank of the $v$-hyperclique tensor is $\Omega(n^v)$; our $\Omega(n^v)$ lower bound for tensor networks strengthens that.
One may wonder about the gap between the bounds of
$\Omega(\binom{n}{n/3})$ and $O(2^n n)$ for the permanent.
As we explain below, our lower bound methods are inherently rank-based
and cannot go beyond $\binom{n}{n/3}$.  A curious point is that it is not immediately clear whether tensor networks can even achieve $O^*(2^n)$ time for the determinant, and we do not know whether or not this is possible.

\subsection{Overview of proof ideas}

\label{sect:ideas}

As a running example in this overview, we consider the $6$-linear forms
$A: \F^{[n] \times [n]} \times \F^{[n] \times [n]} \times \ldots
\times \F^{[n] \times [n]} \rightarrow \F$
taking as input $6$ matrices of size $n \times n$, defined by
\begin{equation}
\label{eq:overview form}
A(X^{(1)}, X^{(2)}, X^{(3)}, X^{(4)}, X^{(5)}, X^{(6)}) = \sum_{i, j, k, \ell \in[n]} X^{(1)}_{ij} X^{(2)}_{ik} X^{(3)}_{i\ell} X^{(4)}_{jk} X^{(5)}_{j\ell} X^{(6)}_{k\ell}\,.
\end{equation}
If $\chi$ is the adjacency matrix of a loopless graph $G$, then $A(\chi, \chi, \chi, \chi, \chi,
\chi)$ counts the number of $4$-cliques in the graph.  Associated with $A$ is the $6$-tensor $T(A)$ of size $n^2 \times n^2 \times \cdots \times n^2$, where each of the $6$ modes is indexed by a pair $(i, j) \in
[n] \times [n]$, and the value at a given position is the coefficient of the corresponding term in $A$.  Concretely,
\[
% Ugly width-fitting-hacks in this eqn, remove all "\!":s to undo
T(\!A)_{i_1j_1, i_2k_2, i_3\ell_3, j_4k_4, j_5\ell_5, k_6\ell_6} \!=\!\!
\begin{cases}
  1 & \!\!\! \text{if $i_1 = i_2 = i_3$, $j_1 = j_4 = j_5$, $k_2 = k_4 = k_6$ $\wedge$ $\ell_3 = \ell_5 = \ell_6$},\\
  0 & \!\!\! \text{otherwise}.
\end{cases}
\]

\subparagraph*{Upper bounds}

Most, but not all, of the families of multilinear maps we consider are
closed under taking Kronecker products.  For instance, consider the
$4$-clique counting form \eqref{eq:overview form} for an $n$-vertex
graph and its associated tensor $T(A)$.  Then for any
$k \ge 1$, the tensor associated with the $4$-clique counting form in
$n^k$-vertex graphs is $T(A)^{\otimes k}$, the $k$-fold Kronecker
product of $T(A)$ with itself.  We write $A^{\otimes k}$ for the
associated map.  With this in mind, it is natural to seek general
constructions that, given an efficient evaluation
of some map $A$, yields an efficient evaluation of $A^{\otimes k}$.

We give such a construction, and show that the cost of the best tensor
network execution for $A^{\otimes k}$ is essentially submultiplicative
in a quantity that we call the \emph{amortized cost} of an execution.
For tensors of order at most $3$, the notion of amortized cost
essentially captures the rank of $T(A)$, but for higher-order tensors,
the amortized cost may be significantly smaller than the rank.
Roughly speaking, the amortized cost of a step in an execution of a
map $A$ is: (i) equal to the normal cost if the operation involves the
contraction of two tensors that both depend on some input variables of
$A$, but (ii) equal to the size of the result if only one of the
tensors involved in the contraction depends on the input variables of
$A$.  A precise definition appears in \cref{sec:submultiplicativity}.
Our general upper bound for the cost of $A^{\otimes k}$ can, somewhat
informally, be stated as follows.

\begin{restatable}[Submultiplicativity of cost, informal statement]{Thm}{GeneralUpperBoundTheorem}
  \label{thm:general upper bound}
  If a multilinear map $A$ has a tensor network execution
  consisting of $s$ steps, each with cost at most $c$ and amortized
  cost at most $a$, then $A^{\otimes k}$ has a tensor network
  execution consisting of at most $k \cdot s$ steps, each with cost at
  most $a^{k-1} \cdot c$.
\end{restatable}

An immediate corollary of this is that we can capture any algorithm
for $A^{\otimes k}$ based on a low-rank decomposition of $T(A)$
(\cref{cor:low rank execution}).  For example, this implies that
tensor networks can multiply $n \times n$ matrices in $O(n^{\omega +
  \epsilon})$ time (\cref{sec:matrix mul upper bound}).

However, returning to our running example form \eqref{eq:overview
  form}, as we explain below the tensor $T(A)$ has rank $n^4$, meaning
that \cref{cor:low rank execution} only yields a trivial upper bound.
This is where the full generality of \cref{thm:general upper bound}
comes in.  Consider the form \eqref{eq:overview form} for graphs on
some constant number $n_0$ of vertices.  As it turns out, we can design
a network and an associated execution for this form, depicted
in \eqref{fig:overview 4 choose 2 execution} and explained in more detail
in the proof of \cref{lem:v choose 2 linear form upper bound},
with an execution of cost $n_0^{2e+3}$ and amortized cost
$n_0^{e+1}$, where $n_0^e$ is the rank of the tensor associated with
$n_0 \times n_0$ matrix multiplication.  Picking $n_0$ to be a large
enough \emph{constant} so that
$e$ is approximately $\omega$, and letting $k$ be such that $n$ is
approximately $n_0^k$, we obtain via \cref{thm:general upper bound} an
$O(n^{\omega +\epsilon + 1})$ time upper bound for \eqref{eq:overview form}.
\slimfigure{0.75}{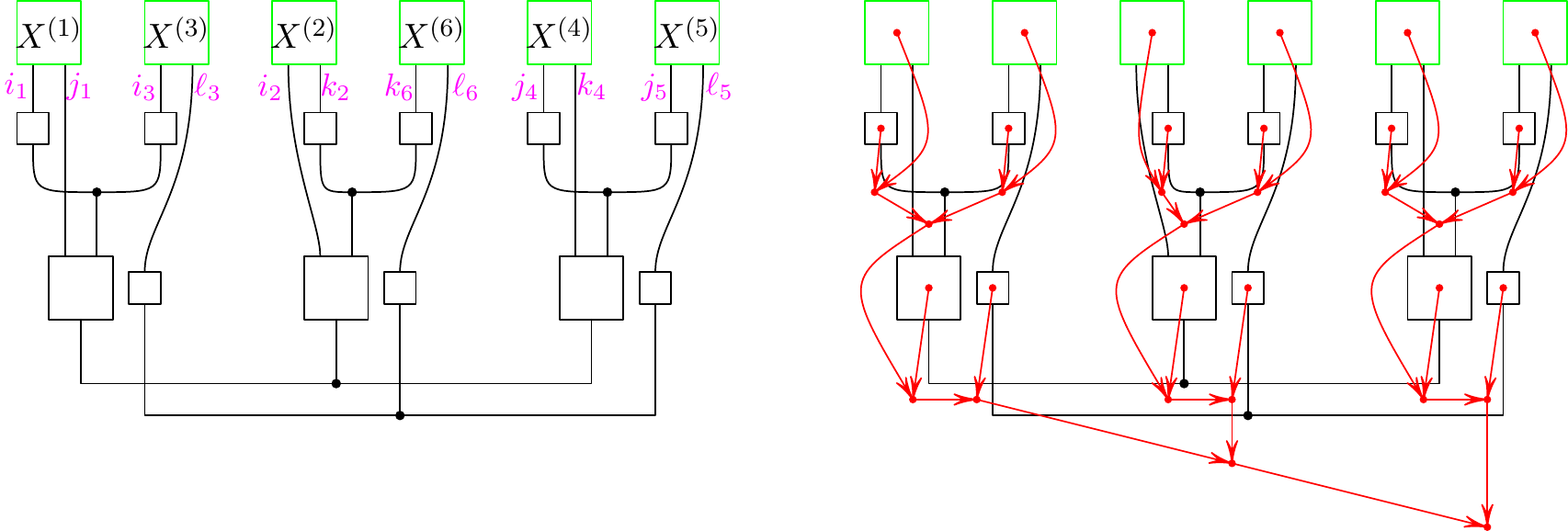}{overview 4 choose 2 execution}

\subparagraph*{Lower bounds}

Just like many other arithmetic complexity lower bounds, our lower bounds
boil down to establishing lower bounds on the rank of certain
matrices.

In order to establish a lower bound on the \emph{rank} of $T(A)$,
we \emph{flatten} it to a matrix and analyze the rank of that
matrix.  There are $2^5$ possible bipartitions of the $6$ modes of
$T(A)$ into two non-empty subsets, and the lower bound on the rank of
$T(A)$ that we obtain is the maximum of the ranks of the
resulting matrices.  Using this method it is easy to establish that
for our example form \eqref{eq:overview form}, the rank of
$T(A) = n^4$.  That this is an upper bound follows from
\eqref{eq:overview form}, and that it is a lower bound follows by
considering the bipartition taking variables $X^{(1)}$ and
$X^{(6)}$ as row indices, and the other $4$ variables as column indices.
The resulting $n^4 \times n^8$ matrix has full rank.

Tensor networks are more versatile and can be more efficient
than low-rank decompositions of $T(A)$.  Nevertheless, we show
limitations on this versatility.  In particular we show that every
tensor network execution for $A$ induces a tree in which the leaves are the inputs of $A$ and all internal vertices have degree $3$.  We call this a \emph{socket tree}.  Each edge in a socket
tree induces a bipartition of the variables and our key technical
lemma is to show that for each such bipartition, the rank of the
corresponding flattening of $T(A)$ is a lower bound on the
cost of the execution that gave rise to the tree.  Thus, to obtain a lower
bound for the cost of a specific execution, we consider the maximum
rank obtained over all
edges of the corresponding socket tree, and to lower bound the cost of
every tensor network execution, we minimize this quantity over all
possible socket trees.  We refer to the resulting quantity as the
\emph{socket width} of $A$, denoted $w(A)$ (formal definition appears in \cref{sec:general lower bound}).  Our general lower bound
can thus be phrased as follows, where $c(A)$ denotes the minimum cost of a tensor network for evaluating $A$ (formal definition appears in \cref{sec:map cost}).

\begin{restatable}{Thm}{GeneralLowerBoundTheorem}
  \label{thm:general lower bound}
  For every multilinear map $A$, it holds that $c(A) \ge w(A)$.
\end{restatable}

Indeed, for our running example \eqref{eq:overview form}, there are
low-width socket trees establishing that $w(A) \le n^3$, see
\eqref{fig:overview socket tree}.  However, that bound is tight: our
$\Omega(n^{\lceil 2 \cdot 4 / 3 \rceil}) = \Omega(n^3)$ lower bound
for the $\binom{4}{2}$-linear form (\cref{thm:lower bounds})
is obtained by proving  that $w(A) \ge n^3$ (\cref{lem:lb-binom-2}) and
appealing to \cref{thm:general lower bound}.
\slimfigure{1.0}{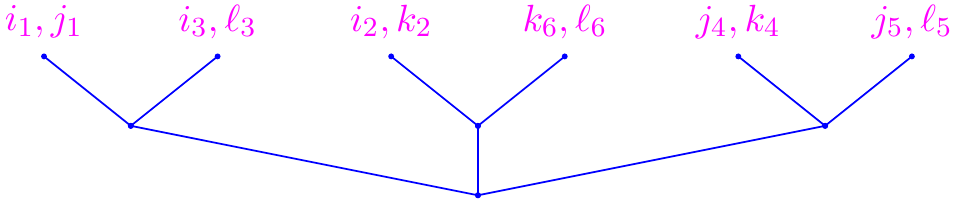}{overview socket tree}

\vspace{-0.5cm}

\subsection{Earlier and related work}
\label{sec:related work}

We now proceed to a more detailed discussion of earlier work.

\subparagraph*{Tensor networks}

The history of tensor networks
(or, alternatively, {\em tensor diagrams} or {\em diagrams}) as
an analytical and computational tool goes back to the \nth{19}-century to the works
by Cayley~\cite{Cayley1875,Cayley1879}, Clebsch~\cite{Clebsch1861},
Clifford~\cite{Clifford1878}, Sylvester~\cite{Sylvester1878},
and Kempe~\cite{Kempe1885,Kempe1892}. The diagrammatic form used here can be
traced back to Penrose~\cite{Penrose1971}. Some early appearances of tensor
diagrams are by Cvitanovi\`c~\cite{Cvitanovic1976},
Cvitanovi\`c and Kennedy~\cite{CvitanovicK1982},
Kuperberg~\cite{Kuperberg1991}, and many others. Surveys of tensor
diagrams can be found in Penrose and Rindler~\cite{PenroseR1986a} and
Landsberg~\cite{Landsberg2012}. Schrijver~\cite{Schrijver2015a} gives
a brief historical account.

A principal deviation from Penrose's notation that we make in this paper
is to work subject to a basis and a corresponding dual basis in each relevant
vector space to avoid distinction between primal and dual spaces.
This in particular enables a concise treatment of hyperedges and saves us
from considering orientation of edges, or the planar layout of edges
in a drawing. That is, we will view a tensor network combinatorially as a
hypergraph with tensors associated at the vertices, and with a subset
of the hyperedges designated to form the
boundary of the network (cf.~\cref{sec:tensor network defs} for the
precise definitions). A yet further conceptual difference is that we view
the execution of a tensor network as a sequence of contractions
of sets of {\em vertices} (tensors) rather than as contractions of
hyperedges (modes). This choice enables us to reduce the size of hyperedges
gradually before eliminating them during an execution, thus
enabling better granularity. For simplicity, we will restrict to
a purely multilinear framework and will not consider sums of networks
although such a study would be possible, and is pursued e.g.~in
Penrose's work~\cite{Penrose1971}.

A large body of existing work in applications (cf.~\cref{sect:nets}) studies
how to efficiently execute a {\em given} tensor network $D$.
Our quest in this paper differs from such studies in that
we study a multilinear map $A$, and seek to
{\em design the most efficient network $D$ that realizes $A$},
or to establish lower bounds for best-possible designs.
In particular, our upper and lower bounds in~\cref{thm:upper bounds}
and~\cref{thm:lower bounds} are over all tensor networks that realize
a particular map $A$ of interest.

\subparagraph*{Computational problems on tensors and tensor networks}

Problems on given tensors and tensor networks are known to be mostly
computationally hard as soon as the setting changes from matrices to
higher-order tensors.
H\aa{}stad~\cite{Hastad1990} showed that computing the rank of
an explicitly given $\ell$-tensor is NP-complete over finite fields
and NP-hard over rationals whenever $\ell\geq 3$. Hillar and Lim extended the latter result to any field containing $\mathbb{Q}$~\cite{HillarL2013}. They also showed that many other tensor problems such as the eigenvalue, singular value and spectral norm decision and approximation problems as well as rank-1 approximation problem for 3-tensors (over $\mathbb{R}$ and in some cases over $\mathbb{C}$) are NP-hard.

The task of finding the value of a given scalar-valued tensor network
is known to be $\#$P-complete (see e.g.~\cite{BacchusDP2003,BiamonteMT2015}).
Similarly, it is NP-hard to find the most efficient sequence of contractions
for a given network~\cite{LamSW1997,PfeiferHV2014}.

\subparagraph*{Tensor networks in applications}

Beyond our present use of tensor networks as a model of computation to efficiently evaluate multilinear maps, tensor networks are used across a broad range of applications. Accordingly, the following should be viewed as mere pointers to further literature on tensor networks, not as an exhaustive listing of all applications of tensor networks. Orus~\cite{Orus2014} gives an introduction to tensor networks in the context of computational physics. Itai and Landau~\cite{ItaiL2010} study quantum computation and quantum algorithms for evaluating tensor networks~\cite{ItaiL2010}. Solomonik and Hoefler study sparse tensor algebra as a model for parallel programming~\cite{SolomonikH2015}. The Holant framework
introduced by Valiant~\cite{Valiant2008} and studied further by
Cai {\em et al.}~\cite{CaiLX2011,CaiGW2016a} involves the study of multilinear
sum--product expressions that can be viewed as tensor networks.
Tensor networks appear naturally in the study of probabilistic graphical
models~\cite{KollerF2009,MoitraW2018,RobevaS2017,WainwrightJ2008},
and in the study of various machine-learning
problems~\cite{CichockiLOPZM2016,CichockiPZLOSM2017}.

\subparagraph*{Bilinear and multilinear complexity}

As was concisely outlined in \cref{sect:ideas}, for bilinear maps our present work reduces to the study of tensor rank of 3-tensors and an extensive body of work on bilinear complexity, with the arithmetic complexity of the matrix multiplication map as the primary driver. For two starting points to this literature, we refer to the monograph by B\"urgisser, Clausen, and Shokrollahi~\cite{Buergisser1997} and to the survey by Pan~\cite{Pan2014}. Our present work can be seen as generalizing this bilinear theory to higher orders of linearity via tensor networks and their executions. The current state of the art for fast matrix multiplication can be found in Le Gall~\cite{LeGall2012,LeGall2014}, Le Gall and Urrutia~\cite{LeGallU2017}, Vassilevska Williams~\cite{VassilevskaWilliams2012}, Cohn and Umans~\cite{CohnU2013}, Cohn, Kleinberg, Szegedy, and Umans~\cite{Cohn2005}, Benson and Ballard~\cite{BensonB2015}, and Huang, Rice, Matthews, and van de Geijn~\cite{HuangRMG2017}.

\subsection{Organization of this paper}

Section~\ref{sect:preliminaries} recalls preliminaries on tensors, multilinear maps, and branchwidth. Section~\ref{sect:examples} reviews the specific multilinear maps that we study in this work and describes for each map its associated tensor. In Section~\ref{sec:tensor network defs}, tensor networks, execution and cost of a tensor network, and cost of a multilinear map are defined. Section~\ref{sec:upper bounds} presents tensor-network algorithms for the multilinear maps introduced in Section~\ref{sect:examples}. In Section~\ref{sec:general lower bound}, a general lower bound on the cost of evaluating a multilinear map using tensor network is obtained. The lower bound is expressed in terms of the socket-width of a multilinear map. In Section~\ref{sec:socket-width lower bound}, lower bounds on socket-width for concrete multilinear maps studied in Sections~\ref{sect:examples} and~\ref{sec:upper bounds} are derived. Finally, Appendix~\ref{sect:exec} gives some background results on minimum-cost executions.

\section{Preliminaries}
\label{sect:preliminaries}

Throughout the paper $[n]$ denotes $\{1,2,\ldots,n\}$ and $\F$ denotes
an arbitrary fixed field.

\subsection{Tensors}
\label{sect:tensors-maps}

This section sets up our notation for tensors and multilinear maps.
We work with tensors and multilinear maps relative to fixed bases for
the respective vector spaces over $\F$.

\subparagraph*{Modes, indexing, and positions}

We will work with the following convention of positioning individual
entries inside a tensor.  Let $E$ be a finite set of {\em
  modes}. Associate with each mode $e\in E$ a finite nonempty {\em
  index set} $J(e)$. In this case we say that $E$ is a set of {\em
  indexed} modes.  The {\em length} of $e$ is $|J(e)|$. A {\em position}
  is an element $j\in \prod_{e\in
  E}J(e)$. Let us write $J(E)=\prod_{e\in E}J(e)$ for the set of all
positions with respect to the indexed modes $E$.  In the special case
that the set of modes $E$ is empty we define the set of positions
$J(E)$ to consist of a single element.

\subparagraph*{Tensors, matrices, vectors, and scalars}

\label{sect:tensor-def}

Let $E$ be a set of indexed modes. A {\em tensor} with respect to $E$ is
a mapping $T:J(E)\rightarrow \F$. Equivalently, we write $T\in\F^{J(E)}$
to indicate that $T$ is a tensor with respect to the indexed modes $E$.
We view the set $\F^{J(E)}$ of all tensors over $E$ as a vector space over $\F$
with addition and scalar multiplication of tensors defined entrywise.
We say that $|E|$ is the {\em order} of the tensor.  A tensor of order
zero is called a {\em scalar}, a tensor of order one is called a {\em
  vector}, and a tensor of order two is called a {\em matrix}.
The {\em volume} of the tensor is $|J(E)|$. The tuple $(|J(e)|:e\in E)$
is the {\em shape} of the tensor. It is convenient to use the
``$\times$''-symbol to punctuate the shape of a tensor; that is,
instead of writing, say $(2,3,4)$ for the shape, we write $2\times 3\times 4$.
For a position $j\in J(E)$ and a tensor $T\in\F^{J(E)}$,
we say that $T_j\in\F$ is the {\em entry} of $T$ at $j$.

A \emph{flattening} of $T$ induced by a bipartition $E_1 \cup E_2 = E$
of the modes of $T$ is a $|J(E_1)| \times |J(E_2)|$ matrix $M$ where, for $j_1 \in J(E_1)$ and $j_2 \in J(E_2)$ we have $M_{j_1,j_2} = T_{j_1j_2}$.
Given two order $\ell$ tensors $S \in \F^{[n_1] \times [n_2] \times \cdots \times [n_\ell]}$ and $T \in \F^{[m_1] \times [m_2] \times \cdots \times [m_\ell]}$, their \emph{Kronecker product} $S \otimes T$ is a tensor in $\F^{[n_1 m_1] \times [n_2 m_2] \times \cdots \times [n_\ell m_\ell]}$ defined by
$$
(S \otimes T)_{m_1(i_1-1)+j_1,m_2(i_2-1)+j_2,\ldots,m_\ell(i_\ell-1)+j_\ell} = S_{i_1,i_2,\ldots,i_\ell} T_{j_1,j_2,\ldots,j_\ell}.
$$

\subsection{Multilinear maps}

Let $E_1,E_2,\ldots,E_\ell,E'$ be pairwise disjoint sets of indexed modes
such that $E_1,E_2,\ldots,E_\ell$ are nonempty. We say that a map
$
A:\F^{J(E_1)}\times\F^{J(E_2)}\times\cdots \times\F^{J(E_\ell)}
      \rightarrow \F^{J(E')}
$
is an $\ell$-{\em linear map} if
$A$ is linear with respect to each of its $\ell$ inputs
individually when the other inputs remain fixed.
In particular, a $1$-linear map is a linear map.
A multilinear map that gives a scalar output is a multilinear {\em form}.
In particular, $A$ is a form if and only if $E'$ is empty.

\subparagraph*{The tensors of a multilinear map}

For an $\ell$-linear map
$
A:\F^{J(E_1)}\times\F^{J(E_2)}\times\cdots \times\F^{J(E_\ell)}\rightarrow \F^{J(E')}\,,
$
we define two slightly different tensors $T(A)$ and $\hat{T}(A)$.  Both are indexed by $J(E_1 \cup E_2 \cup \ldots \cup E_\ell \cup E')$ and at position $j_1j_2\ldots j_\ell j'$ take the value
\[
T(A)_{j_1j_2\ldots j_\ell j'} = \hat{T}(A)_{j_1j_2 \ldots j_\ell j'} =
A\bigl(e^{(j_1)},e^{(j_2)},\ldots,e^{(j_\ell)}\bigr)_{j'}\,,
\]
where $e^{(j_i)} \in \F^{J(E_i)}$ denotes the tensor with a $1$ in
position $j_i$ and $0$s in all other position.  The difference between
$T(A)$ and $\hat{T}(A)$ is their shape.  The shape of $T(A)$ is
$|J(E_1)| \times |J(E_2)| \times \cdots \times |J(E_\ell)| \times
|J(E')|$,
except if $A$ is a form in which case the $|J(E')|$ part is omitted.
Thus $T(A)$ is of order $\ell+1$ (or $\ell$ if $A$ is a form).  The
shape of $\hat{T}(A)$ is $(|J(e)|: e \in E_i, i \in [\ell+1])$, thus
its order is $|E_1| + |E_2| + \cdots + |E_\ell| + |E'|$.

In other words, each mode of $T(A)$ corresponds to one of the inputs
of $A$, or the output.  These inputs are in turn sets of indexed modes
so may contain more ``fine-grained'' structure, but this information
is lost at the level of granularity of $T(A)$.  When working with
tensor networks for evaluating $A$, we need to keep track of the
fine-grained mode structure because this is in many cases what allows
us to construct efficient algorithms, hence in most parts of the paper
we are more interested in the tensor $\hat{T}(A)$ which contains this
fine-grained structure.

On the other hand, $\hat{T}(A)$ does not contain information about
which modes are grouped together to form the inputs and output of $A$,
and this information is also important.  This leads us to the notion
of sockets, defined next.

\subparagraph*{Sockets}

Let us study the tensor $\hat{T}(A)$ with respect to the map~$A$.
We say that the modes in
$E_1\cup E_2\cup\cdots\cup E_\ell$ are the {\em input} modes of $\hat{T}(A)$,
and the modes in $E'$ are the {\em output} modes of $\hat{T}(A)$ with respect to $A$.
Let us say that $E_1, \ldots, E_\ell$ are the
{\em input sockets} of $\hat{T}(A)$ with respect to $A$. Similarly,
$E'$ is the {\em output socket} in $\hat{T}(A)$ with respect to $A$.
In particular, the output socket is empty if and only if $A$ is a form. To describe a socketing of the modes of a tensor, it is convenient to
use parentheses to group a ``$\times$''-punctuated shape of a tensor
into sockets, see also~\cref{sect:tensor-def}.

Let $\hat{T}$ be a tensor over a set of indexed modes $E$.
Any tuple $\mathcal{E} = (E_1,E_2,\ldots,E_\ell,E')$ of subsets of $E$
that partitions $E$ with $E_1,E_2,\ldots,E_\ell$ nonempty
defines an $\ell$-linear map
$A_{\mathcal{E}}(\hat{T}):\F^{J(E_1)}\times\F^{J(E_2)}\times\cdots \times\F^{J(E_\ell)}
\rightarrow\F^{J(E')}$ with $\hat{T}(A_{\mathcal{E}}(\hat{T}))=\hat{T}$. In this case the tuple
$(E_1,E_2,\ldots,E_\ell)$ gives the input sockets of $T$ and $E'$
is the output socket of $\hat{T}$ with respect to $A_{\mathcal{E}}(\hat{T})$.

We thus conclude that two multilinear maps $A_1,A_2$
may have the same base tensor $\hat{T}(A_1)=\hat{T}(A_2)$, and from a tensor $\hat{T}$ one may
obtain different multilinear maps by varying how the modes of $\hat{T}$
are assigned to input and output sockets.

\subparagraph*{The form of a multilinear map}

Let $A$ be a multilinear map with a nonempty output socket.
We can turn $A$ into a multilinear form $F(A)$ by turning
its output socket into an input socket.
Let us say that $F(A)$ is the {\em multilinear form} of $A$.
We also set $F(A)=A$ when $A$ is a multilinear form.

\subsection{Branch decompositions and branchwidth}

A {\em branch decomposition} of a (hyper)graph $G$ consists of
(i) a tree $T$ whose every vertex has degree either one or three, and
(ii) a bijection $\pi$ between the (hyper)edge set of $G$ and the set of
vertices of degree one in $T$. Deleting an edge $e\in E(T)$ from $T$
partitions $T$ into two subtrees $T_1$ and $T_2$ that via $\pi$ give rise
to a partition of the (hyper)edges of $G$ into two sets $E_1$ and $E_2$.
The {\em width} $w(e)$
of the partition induced by $e$ is the number of vertices of $G$ that are
incident to at least one (hyper)edge in $E_1$ and at least one (hyper)edge
in $E_2$. The {\em width} of the branch decomposition $(T,\pi)$ is the maximum
of the widths $w(e)$ for $e\in E(T)$. The {\em branchwidth} $\bw(G)$ of $G$
is the minimum width of a branch decomposition of $G$.

For graphs, we recall the following upper bound on branchwidth.

\begin{Claim}[Robertson and Seymour~\cite{RobertsonS1991}]
  For every $n \ge 3$, $\bw(K_n) = \lceil 2n/3 \rceil$.  Consequently, $\bw(G) \le \lceil 2 |V(G)|/3 \rceil$ for every graph $G$.
\end{Claim}

\section{Examples of multilinear maps}
\label{sect:examples}

This section reviews the specific multilinear maps that we study in this work.
Together with each map we describe its associated
tensor and a socketing of the tensor that realizes the map.

\subsection{Discrete Fourier transform}

\label{sect:dft}

Let $n\geq 2$ be an integer and let $\rho\in\F$ be a primitive
$n^\text{th}$ root of unity in the field $\F$.%
\footnote{That is, $\rho$ satisfies $\rho^n=1$ and $\rho^{n/k}\neq 1$
for all integer divisors $k\geq 2$ of $n$.}{}
Define the $n\times n$ symmetric matrix $\Phi$ with
entries $\Phi_{i,j}=\rho^{(i-1)(j-1)}$ for all $i,j\in [n]$.
The {\em discrete Fourier transform (DFT) of order $n$ at $\rho$} is the
linear map $L:\F^{[n]}\rightarrow \F^{[n]}$ defined for all $x\in\F^{[n]}$
by the matrix-vector multiplication $L(x)=\Phi x$.
In particular, the matrix $\Phi$ is the tensor
associated with $L$. The matrix $\Phi$ has two modes,
namely its rows and columns.
To realize $L$, take the columns of $\Phi$ as the input socket, and the rows
as the output socket.

\subsection{Determinant and permanent}

Let us write $S_n$ for the symmetric group of all permutations
$\sigma:[n]\rightarrow [n]$. Let $c(\sigma)$ be the number of
cycles in the cycle decomposition of $\sigma$,
where each fixed point of $\sigma$ is counted as a cycle.
Further standard examples of multilinear operators include the
determinant and permanent
\begin{align}
\label{eq:det and per}
\det A &=
\sum_{\sigma\in S_n}(-1)^{n-c(\sigma)}\prod_{i\in[n]}a_{i,\sigma(i)} &
\per A &=
\sum_{\sigma\in S_n}\prod_{i\in[n]}a_{i,\sigma(i)}\,,
\end{align}
both of which are $n$-linear in the $n$ columns (or the $n$ rows) of the
matrix $A=(a_{ij})_{i,j\in[n]}$.
The determinant and permanent are associated with order $n$ tensors
$\hat{T}^{(\det)}, \hat{T}^{(\per)} \in \F^{[n] \times [n] \times \cdots \times [n]}$
defined for all $j = (i_1,i_2,\ldots,i_n) \in [n] \times [n] \times \cdots \times [n]$ by

\begin{align}
\label{tensor:det}
\hat{T}^{(\det)}_{j}&=
\begin{cases}
(-1)^{n-c(j)} & \text{if } j \in S_n\\
0 & \text{otherwise}
\end{cases} &
\hat{T}^{(\per)}_{j}&=
\begin{cases}
1 & \text{if } j \in S_n\\
0 & \text{otherwise}.
\end{cases}
\end{align}
To realize the determinant and permanent using \eqref{tensor:det}, we socket $\hat{T}^{(\det)}$ and $\hat{T}^{(\per)}$ with the rows of $A$.
The determinant and permanent tensors are not closed under taking Kronecker products, because the Kronecker product of two tensors has the same order as its input tensors, but for each $n$ there is exactly one determinant (permanent) tensor of order $n$.

\subsection{Matrix multiplication}
\label{sect:mm}

Let $n$, $r$, and $m$ be positive integers.
Perhaps the most fundamental example of a bilinear map is the map
that multiplies
an $n\times r$ matrix $A=(A_{ij})_{i\in[n],\,j\in [r]}$ with
an $r\times m$ matrix $B=(B_{ij})_{i\in[r],\,j\in [m]}$
to obtain the $n\times m$ product matrix $C=(C_{ij})_{i\in [n]\,,j\in [m]}$
defined for all $i\in[n]$ and $j\in [m]$ by
\begin{equation}
\label{eq:mm}
C_{ij}=\sum_{k\in[r]} A_{ik}B_{kj}\,.
\end{equation}
It is natural to view the input $A\in\F^{[n]\times [r]}$ as a 2-tensor,
and similarly so for the input $B\in\F^{[r]\times [m]}$, and the
output $C\in\F^{[n]\times [m]}$. Thus, \eqref{eq:mm} is naturally
associated with the 6-tensor
$\hat{T} \in\F^{[n]\times [r]\times [r]\times [m]\times [n]\times [m]}$ with
entries defined for all $i_1,i_2\in [n]$, $j_1, j_2\in[m]$, and
$k_1,k_2\in[r]$ by
\begin{equation}
\label{eq:mm-t}
\hat{T}_{i_1\!k_1\!k_2j_1\!i_2j_2}=
\begin{cases}
1 & \text{if $i_1=i_2$ and $j_1=j_2$ and $k_1=k_2$};\\
0 & \text{otherwise}.
\end{cases}
\end{equation}
To realize \eqref{eq:mm} using \eqref{eq:mm-t}, we can use the socketing
grouped by parentheses in
$([n]\times [r])\times([r]\times [m])\times([n]\times [m])$, where the first
two groups are the two input sockets corresponding to $A$ and $B$,
and the last group is the output socket corresponding to $C$.

Let us write $\bra n,r,m\ket$ as a shorthand for the tensor \eqref{eq:mm-t}.
From \eqref{eq:mm-t} and the definition of the Kronecker product
it is immediate that matrix-multiplication tensors satisfy
\begin{equation}
\label{eq:mm-kron}
\bra n_1,r_1,m_1\ket
\otimes
\bra n_2,r_2,m_2\ket
=
\bra n_1n_2,r_1r_1,m_1m_2\ket\,.
\end{equation}
That is, matrix multiplication tensors are closed under taking  Kronecker products.

\subsection{Group algebra product}

Let $(\mathbb{A},+)$ be an Abelian group of order $n$.
Another fundamental example of a bilinear map is to convolve
two vectors $f\in\F^{\mathbb{A}}$ and $g\in\F^{\mathbb{A}}$ according to the
group operation of $\mathbb{A}$ to obtain
the vector $h=f*g\in\F^{\mathbb{A}}$, defined for all $k\in \mathbb{A}$ by
\begin{equation}
\label{eq:conv}
h_k=\sum_{j\in \mathbb{A}}f_{(k-j)}\,g_j\,.
\end{equation}
The map \eqref{eq:conv} is associated with the $3$-tensor
$\hat{T}\in\F^{\mathbb{A}\times\mathbb{A}\times\mathbb{A}}$ defined for all
$i,j,k\in \mathbb{A}$ by
\begin{equation}
\label{eq:conv-t}
\hat{T}_{ijk}=
\begin{cases}
1 & \text{$i+j=k$};\\
0 & \text{otherwise}.
\end{cases}
\end{equation}
A socketing of the three modes of \eqref{eq:conv-t} that
realizes \eqref{eq:conv}
is to take the first two modes as two input sockets corresponding
to $f$ and $g$, respectively, and to take the last mode as the
output socket corresponding to $h$. The vector space
$\F^\mathbb{A}$ equipped with the convolution product $*$ is the
{\em group algebra} $\F[\mathbb{A}]$.

\subsection{Kruskal operator}
\label{sect:kruskal}

Let $n_1,n_2,\ldots,n_\ell$ and $r$ be positive integers.
Matrix multiplication generalizes naturally to the $\ell$-linear task
of multiplying $\ell$ matrices
$A^{(1)}\in\F^{[n_1]\times [r]}$,
$A^{(2)}\in\F^{[n_2]\times [r]}$,
$\ldots$,
$A^{(\ell)}\in\F^{[n_\ell]\times [r]}$
to obtain the order-$\ell$ tensor
$Y\in\F^{[n_1]\times [n_2]\times \cdots\times [n_\ell]}$
defined for all
$(i_1,i_2,\ldots,i_\ell)\in[n_1]\times[n_2]\times\cdots\times[n_\ell]$ by
\begin{equation}
\label{eq:kruskal}
Y_{i_1i_2\cdots i_\ell}
=\sum_{j\in [r]} A^{(1)}_{i_1j}A^{(2)}_{i_2j}\cdots A^{(\ell)}_{i_\ell j}\,.
\end{equation}
This operator is known as the {\em Kruskal operator}~\cite{Kruskal1977,Kolda2006}%
\footnote{Kolda~\cite{Kolda2006} calls this operator the Kruskal operator. Kruskal~\cite{Kruskal1977} studied the case $\ell=3$ in particular.}{}
of the matrices $A^{(1)},A^{(2)},\ldots,A^{(\ell)}$.

The map \eqref{eq:kruskal} is associated with the $3\ell$-tensor
\[
\hat{T}\in\F^{%
[n_1]\times [r]\times
[n_2]\times [r]\times
\cdots \times
[n_\ell]\times [r]\times
\cdots \times
[n_1]\times [n_2]\times\cdots\times [n_\ell]}
\]
defined for all
$i_1,i_1'\in [n_1]$,
$i_2,i_2'\in [n_2]$,
$\ldots$
$i_\ell,i_\ell'\in [n_\ell]$
and
$j_1,j_2,\ldots,j_\ell\in [r]$
by
\begin{equation}
\label{eq:kruskal-t}
\hat{T}_{i_1j_1i_2j_2\cdots i_\ell j_\ell i_1'i_2'\cdots i_\ell'}=
\begin{cases}
1 & \text{if $i_1=i_1'$, $i_2=i_2'$, $\ldots$, $i_\ell=i_\ell'$ and $j_1=j_2=\ldots=j_\ell$};\\
0 & \text{otherwise}.
\end{cases}
\end{equation}
A socketing of \eqref{eq:kruskal-t} that realizes \eqref{eq:kruskal}
is to take each of the $\ell$ pairs of modes
$[n_1]\times[r]$, $[n_2]\times[r]$, $\ldots$, $[n_\ell]\times [r]$
as an input socket and the final $\ell$ modes
$[n_1]\times [n_2]\times\cdots\times[n_\ell]$ as the output socket.

Let us write $\bra n_1,n_2,\ldots,n_\ell|r\ket$ for the tensor
\eqref{eq:kruskal-t}. Analogously to the closure property \eqref{eq:mm-kron}
for matrix multiplication tensors, we observe that
\begin{equation}
\label{eq:kruskal-kron}
\bra n_1,n_2,\ldots,n_\ell|r\ket
\otimes
\bra n_1',n_2',\ldots,n_\ell'|r'\ket
=
\bra n_1n_1',n_2n_2',\ldots,n_\ell n_\ell'|rr'\ket\,.
\end{equation}
That is, Kruskal operator tensors are closed under taking Kronecker products.
Furthermore, we observe that $\bra n,r,m\ket=\bra n,m|r\ket$. That is, the
matrix multiplication tensor \eqref{eq:mm-t} is the special case
of the Kruskal operator tensor \eqref{eq:kruskal-t} when $\ell=2$.

\subsection{Homomorphism-counting forms}

\label{sect:binom}

Further multilinear operators arise
naturally by algebraization of combinatorial problems. For example,
the linear form of the matrix multiplication map
\begin{equation}
\label{eq:3-choose-2}
\sum_{i,j,k\in [n]}A_{ij}B_{jk}C_{ki}
\end{equation}
can be used to count the number of triangles in graph, and the form
\begin{equation}
\label{eq:4-choose-2}
\sum_{i,j,k,\ell\in [n]}
A_{ij}B_{ik}C_{i\ell}D_{jk}E_{j\ell}F_{k\ell}
\end{equation}
can be used to count the number of occurrences of any 4-vertex subgraph in
a graph by varying the six $n\times n$ matrices
$A,B,C,D,E,F\in\F^{[n]\times [n]}$. A more complicated
variant takes as input four $3$-tensors
$A, B, C, D \in\F^{[n]\times [n]\times [n]}$
of shape $n\times n\times n$ and considers the linear form
\begin{equation}
\label{eq:4-choose-3}
\sum_{i,j,k,\ell\in [n]}
A_{ijk}B_{ik\ell}C_{ij\ell}D_{jk\ell}\,.
\end{equation}
An $n^{4-\delta}$ time algorithm for the associated trilinear map $(A, B, C) \mapsto D$ would imply an algorithm for the Max 3-Sat problem with running time $(2-\delta')^n$ \cite{Williams2005}.

The forms \eqref{eq:3-choose-2}, \eqref{eq:4-choose-2}, and \eqref{eq:4-choose-3} are all special cases of what we call a \emph{homomorphism-counting form}
defined by a hypergraph $P$, or, succinctly, a $P$-{\em linear form}.
In more precise terms, let $P$ be a $k$-uniform hypergraph on $v \ge k$ vertices and write $\binom{[v]}{k}$ for the set of $k$-element subsets of $[v]$.  For each hyperedge $S=\{i_1,i_2,\ldots,i_k\} \in E(P)\subseteq \binom{[v]}{k}$ of $P$, let $X^{(S)}$ be an input tensor of shape $[n]^S$.  The \emph{$P$-linear form} of {\em order $n$} is the form
\begin{equation} \label{eq:homomorphism_form}
\sum_{\sigma \in [n]^{V(P)}} \prod_{S \in E(P)} X^{(S)}_{\sigma|_S},
\end{equation}
where $\sigma|_S$ is the restriction of $\sigma$ to the $k$ indices in $S$.
For example, the forms \eqref{eq:3-choose-2}, \eqref{eq:4-choose-2}, and \eqref{eq:4-choose-3} are the $P$-linear forms corresponding to $P$ being a triangle, a $K_4$, and a complete $3$-uniform hypergraph on $4$ vertices, respectively.
It is immediate that \eqref{eq:homomorphism_form} is an $|E(P)|$-linear map.

The map~\eqref{eq:homomorphism_form} is associated with a tensor $\hat{T}$ of order $k \cdot |E(P)|$ whose modes are $M =
\{\,(S,i)\,|\,S \in E(P), i \in S \,\}$ with index sets
$J((S,i)) = [n]$. The entries of $\hat{T}$ are defined by
\begin{equation}
\label{eq:linear-forms-t}
\hat{T}_j=
\begin{cases}
1 & \text{if $\exists \sigma \in [n]^{V(P)}$ such that $j_{(S,i)}=\sigma_i$ for every $(S,i) \in M$};\\
0 & \text{otherwise}.
\end{cases}
\end{equation}
A socketing of \eqref{eq:linear-forms-t} that realizes \eqref{eq:homomorphism_form} is given by one input socket for each hyperedge $S \in E(P)$ such that
the socket contains the $k$ modes $(S,i)$ for $i \in S$.

Let us observe that the tensors of the forms \eqref{eq:4-choose-2} and
\eqref{eq:4-choose-3} are actually the same -- they are both of order
$12$, have volume $n^{12}$, and are in fact equal to the outer product
of the $n \times n \times n$ identity tensor with itself $4$ times
after renaming of modes.  However, due to the difference in
socketing, the forms are computationally very different.
We show in \cref{sec:socket-width lower bound} that while there
are non-trivial tensor network algorithms for evaluating
\eqref{eq:4-choose-2}, no such algorithms exist for
\eqref{eq:4-choose-3}.

Let us write $\bra n\ket_P$
for the tensor $\hat{T}$ as defined
in \eqref{eq:linear-forms-t}. Analogously to the closure
properties \eqref{eq:mm-kron} and \eqref{eq:kruskal-kron}, we observe the
closure property
\begin{equation}
  \label{eq:binom-kron}
  \bra n\ket_P \otimes \bra n'\ket_P = \bra nn'\ket_P\,.
\end{equation}

\section{Tensor networks}
\label{sec:tensor network defs}

\subsection{Networks}

A {\em network} (or {\em diagram})
consists of a finite set $V$ of {\em vertices},
a finite set $E$ of {\em hyperedges},
an {\em incidence relation} $I\subseteq V\times E$, and
a {\em boundary} $B\subseteq E$.
A network is {\em nondegenerate} if every hyperedge is incident to at least
one vertex. In what follows we assume that all networks are nondegenerate.
A hyperedge $e\in E$ is a {\em loop} if $e\notin B$ and $e$ is incident
to exactly one vertex.

For a vertex $v\in V$, let us write $I(v)=\{e\in E:(v,e)\in I\}$
for the set of hyperedges incident to $v$. Dually, for an hyperedge $e\in E$, let
us write $I(e)=\{v\in V:(v,e)\in I\}$ for the set of vertices incident to $e$.
For a network $D$, we write $V(D)$, $E(D)$, $I(D)$, and $B(D)$ to
refer to the vertices of $D$, the hyperedges of $D$, the incidence relation
of $D$, and the boundary of $D$, respectively.

\subparagraph*{Induced networks}
For a network $D$ and a nonempty subset $W\subseteq V(D)$,
the {\em induced} network $D[W]$ consists of the vertices in $W$ together with the hyperedges of
$D$ that are incident to at least one vertex in $W$; the boundary of
$D[W]$ consists of all hyperedges that are at the boundary of $D$
or incident to a vertex outside $W$.  Formally,
\begin{equation}
\label{eq:induced}
\begin{aligned}
  V(D[W])&=W\,,\\
  E(D[W])&=\{e\in E(D):\text{$\exists w \in W$ s.t.~$(w,e)\in I(D)$}\}\,,\\
  I(D[W])&=I(D)\cap\left(V(D[W])\times E(D[W])\right)\,,\\
  B(D[W])&=\left(B(D)\cap E(D[W])\right)\cup
  % Ugly space-fiddling hack to get equation number to fit, remove "\!" to undo
  \{e\in E(D[W])\!:\text{$\exists v \in V(D) \!\setminus\! W$ s.t.~$(v,e)\in I(D)$}\}.
\end{aligned}
\end{equation}
For a vertex $v\in V$, we abbreviate $D[v]=D[\{v\}]$.
Note that the boundary of $D[v]$ consists of all non-loop
hyperedges incident to $v$ in $D$.

\subsection{Tensor networks}

Let $D$ be a network.
We \emph{index} $D$ by associating with each hyperedge $e \in E$ an \emph{index set} $J(e)$ of size $\ell(e)$.  Induced networks inherit indexing by restriction.
Next we associate with each vertex $v\in V$ a tensor
$T(v)\in\F^{J(I(v))}$. We say that $D$ equipped with the tensors
$(T(v))_{v \in V}$ is a {\em tensor network}.

The \emph{value} of a tensor network $D$, or the tensor \emph{represented by} $D$, is a tensor $T(D) \in \F^{J(B)}$, defined for all $i \in J(B)$ by
\begin{equation}
\label{eq:tensor-rep}
T(D)_i=
\sum_{j\in J(E(D)\setminus B)}
\prod_{v\in V}T(v)_{ij}\,.
\end{equation}
Observe that in \eqref{eq:tensor-rep}
the positions $i$ and $j$ together identify a unique entry of $T(v)$
by projection to $J(I(v))$.
We also observe that the value of a tensor network with an empty boundary is a scalar.

\subsection{Contracting tensors}

Let $D$ be a tensor network and let $W\subseteq V(D)$ be a nonempty set
of vertices. Let $w$ be a new element not in $V$.
We may {\em contract} $W$ in $D$ to obtain a tensor network $D/W$ by replacing the sub-network $D[W]$ in $D$ with the single vertex $w$ whose associated tensor $T(w)$ is the tensor represented by $D[W]$.  Formally,
\begin{equation}
\label{eq:contraction}
\begin{aligned}
V(D/W)&=(V(D)\setminus W)\cup\{w\}\,,\\
E(D/W)&=E(D)\setminus\left(E(D[W])\setminus B(D[W])\right)\,,\\
I(D/W)&=\left(I(D)\setminus I(D[W])\right)\cup\{(w,e):e\in B(D[W])\}\,,\\
B(D/W)&=B(D)\,,\\
T(w)&=T(D[W])\,.
\end{aligned}
\end{equation}
The {\em cost} of contracting $W$ in $D$ is
$c(D,W)=\prod_{e\in E(D[W])}|J(e)|$.
The value of a tensor network is invariant under contraction, i.e.,
for all nonempty $W\subseteq V(D)$ it holds that $T(D)=T(D/W)$ (see
\cref{lem:invariance} for a proof).

\subsection{Execution and cost of a tensor network}

To compute the tensor $T(D)$ from a given tensor network $D$,
we may proceed by a sequence of contractions on $D$. Such a process is called
executing $D$, and the cost of $D$ is the cost of a minimum-cost execution of $D$. We proceed with the details.

Let $D=D_0$ be a tensor network with at least one tensor.
For $k=1,2,\ldots,t$, select a nonempty
subset $W_{k-1}\subseteq V(D_{k-1})$ such that $W_{k-1}$ has at least two
tensors or consists of a single tensor with a loop.
Set $D_k=D_{k-1}/W_{k-1}$ and observe that the number of tensors and/or modes
decreases by at least one in the contraction.
Suppose that $D_t$ is loopless and consists of a single tensor.
We say that such a sequence of contractions is an {\em execution} of $D$
in $t$ {\em steps}.
The {\em cost} of the execution is
$\max_{k=1,2,\ldots,t} c(D_{k-1},W_{k-1})$.
The cost of an execution in zero steps is defined to be $0$.

It is immediate that $D$ has at least one execution and every
execution consists of at most $2|V(D)|-1$ steps.  By invariance under
contractions, we have $T(D_t) = T(D)$.
The {\em cost} $c(D)$ of $D$ is the cost of a minimum-cost execution of $D$.

An execution of $D$ of cost $c(D)$ immediately translates into an
algorithm that computes $T(D)$ using $O(c(D)|V(D)|)$ arithmetic
operations in $\F$, since the contraction step $D_{k} =
D_{k-1}/W_{k-1}$ takes $O(c(D_{k-1}, W_{k-1})) \le c(D)$ time to
evaluate, and there are $O(V(D))$ steps.

\begin{restatable}{Lem}{BinaryExecutionLemma}
\label{lem:binary}
Let $D$ be a tensor network.
There exists a minimum-cost execution of $D$ such that each contracted
set has size at most two. Furthermore, if $D$ is loopless,
we can assume that each contracted set has size exactly two.
\end{restatable}

In what follows we restrict to consider loopless $D$ only.
Thus while a general execution may contract arbitrary vertex sets in
$D$ in each step, we may assume without loss of generality that
the minimum-cost execution has the structure of a rooted binary tree,
whose leaves are the vertices of the tensor network, and each internal
vertex is the tensor obtained by contracting its two children.

\subsection{Cost of a multilinear map}
\label{sec:map cost}

Let us now define the cost of a multilinear map via the minimum-cost
tensor networks (and socketing) for evaluating the map.
That is, the cost of a multilinear map is defined in
terms of the best tensor network that implements the map.
In more precise terms, let
\[
A:\F^{J(E_1)}\times\F^{J(E_2)}\times\cdots \times\F^{J(E_\ell)}
      \rightarrow \F^{J(E')}
\]
be an $\ell$-linear map. Consider the tensor $\hat{T}(A)$ of $A$ and the
associated input sockets $E_1,E_2,\ldots,E_\ell$ and the output socket
$E'$. Let $D^*$ be an arbitrary tensor network such that $T(D^*)=\hat{T}(A)$ and
the boundary satisfies $B(D^*) =E_1\cup E_2\cup\cdots\cup E_\ell\cup E'$.
Modify the network $D^*$ as follows.
For each $k=1,2,\ldots,\ell$, introduce a new vertex to $D^*$,
make the new vertex incident to each of the modes in the input socket $E_k$,
and associate the new vertex with a tensor $X^{(k)}\in\F^{J(E_k)}$.
Remove the modes $E_1\cup E_2\cup\cdots\cup E_\ell$ from the boundary of $D^*$.
Let us denote the resulting network by $D$ and call the introduced $\ell$
new vertices the {\em socket vertices} of $D$.
We observe that $B(D)=E'$ and $A(X^{(1)},X^{(2)},\ldots,X^{(\ell)})=T(D)$.
Furthermore, the cost $c(D)$ is independent of the value of
$X^{(k)}\in\F^{J(E_k)}$ for $k=1,2,\ldots,\ell$.
We say that $D$ is a \emph{realization} of $A$ if it can be obtained
from $A$ by this process, and write $\mathscr{D}(A)$ for the set of all
tensor network realizations $D$ of $A$.

The {\em cost} of the map $A$ is
$c(A)=\min_{D\in\mathscr{D}(A)} c(D)$.
In particular, we observe that the minimum exists since the cost of a
tensor network is a nonnegative integer and the family $\mathscr{D}(A)$
is nonempty.

\section{Upper bounds on cost}
\label{sec:upper bounds}

This section presents tensor-network algorithms for the maps
introduced in \S\ref{sect:examples}. We start with our key technical result
that cost is submultiplicative (Theorem~\ref{thm:general upper bound},
stated formally as Theorem~\ref{thm:submul-am}),
which then enables us to represent essentially the fastest known algorithms
using tensor networks, and, in the case of $P$-linear forms, also
improve on earlier work as reviewed in \S\ref{sect:our}.

\subsection{Submultiplicativity of cost}
\label{sec:submultiplicativity}

Let $E_1,E_2,\ldots,E_\ell,E'$ be pairwise disjoint sets of indexed modes
such that $E_1,E_2,\ldots,E_\ell$ are nonempty. Let
\[
A:\F^{J(E_1)}\times\F^{J(E_2)}\times\cdots \times\F^{J(E_\ell)}
      \rightarrow \F^{J(E')}
\]
be an $\ell$-linear map.
For a positive integer $k$, we
define the $\ell$-linear map $A^{\otimes k}$ such that its tensor satisfies $T(A^{\otimes k})=T(A)^{\otimes k}$.  Then
\[
A^{\otimes k}:\F^{J(E_1)^k}\times\F^{J(E_2)^k}\times\cdots \times\F^{J(E_\ell)^k}
      \rightarrow \F^{J(E')^k}.
\]
Note that $T(A^{\otimes k}) = T(A)^{\otimes k}$ is the $k$-fold
Kronecker product of $T(A)$ with itself -- that is, it has the same order,
but the index sets are larger -- whereas $\hat{T}(A^{\otimes k})$ is
the $k$-fold outer product of $\hat{T}(A)$ with itself -- that is, its
index sets have the same sizes, but its order is $k$ times larger.

Let $D$ be a diagram that realizes $A$ and let $\td$ be an execution
tree for $D$.
For each internal vertex $x$ in $\td$ (that is, a vertex obtained by contraction),
define the {\em amortized cost} of $x$ by splitting into the following
three cases:
\begin{enumerate}[(i)]
\item
if neither of the two subtrees of $x$ contains a socket vertex,
the amortized cost of $x$ is $1$;
\item
if exactly one of the subtrees of $x$,
say, the subtree rooted at $y$ (where $x$ and $y$ are adjacent in $\td$),
contains at least one socket vertex, the amortized cost of
$x$ is the maximum of the volume of the tensor at $x$ and the volume of
the tensor at $y$;\footnote{Here, it is crucial to note that the volume of the other subtree rooted at $x$, only containing non-socket vertices, does not contribute directly to the amortized cost of $x$.}
\item
if both of the subtrees of $x$ contain at
least one socket vertex, the amortized cost of $x$ is the
cost of the contraction to obtain $x$.
\end{enumerate}

The {\em amortized cost} $a(\td)$ of $\td$ is the maximum of the amortized
costs of the internal vertices of $\td$. Since the amortized cost of each
internal vertex of $\td$ is at most its cost, we have $a(\td)\leq c(\td)$.
Furthermore, we observe that the amortized cost of $x$ in case (ii) above
may be strictly less than the cost of the contraction to obtain $x$.
In particular, in (ii) the amortized cost is defined
{\em not by the cost of the contraction but rather by volume}.
This is because in a $k^{\text{th}}$ Kronecker power we can amortize
the cost of the aggregate transformation in case (ii) not with
a single contraction but with a sequence of $k$ contractions.
This observation will form the heart of the proof of
Theorem~\ref{thm:general upper bound}.

Before proceeding with the proof, let us illustrate the
key ideas in visual terms. Let us start with the three illustrations in
\eqref{fig:diag-base}.
\slimfigure{1.0}{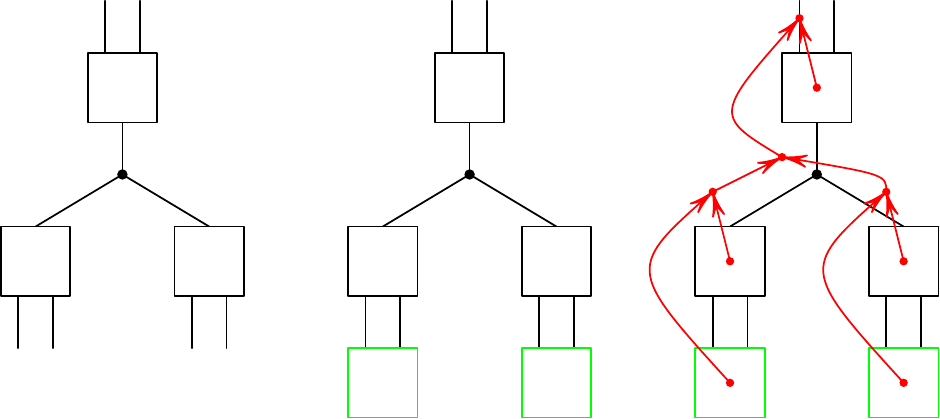}{diag-base}
Suppose the leftmost network in \eqref{fig:diag-base}
is socketed so that the two modes at the top form the output socket,
and the four modes at the bottom form two input sockets so that modes in the
same socket are incident to the same vertex.
In the middle in \eqref{fig:diag-base}, we have adjoined two socket
vertices to the input sockets to obtain
a realization $D$. On the right in \eqref{fig:diag-base}, we display an
execution tree $\td$ for $D$. Observe that the bottom-most internal vertices
of $\td$, and the top-most internal vertex of $\td$, have type (ii). The
internal vertex in the center has type (iii). (There are no internal vertices
of type (i).) Supposing that all the modes have length at least $2$,
we also observe that the vertices of type (ii) have
amortized cost strictly less than their contraction cost.

Let us now consider the $k^{\text{th}}$ power
of \eqref{fig:diag-base} visually, for $k=4$:

\slimfigure{0.44}{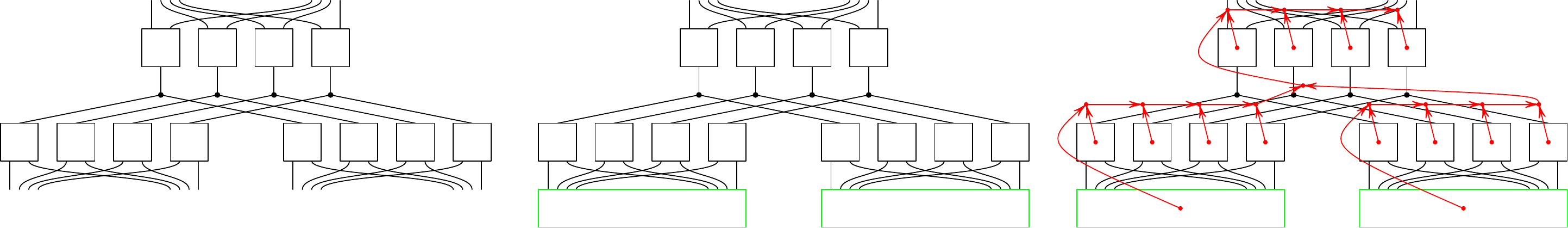}{diag-base-kron}

The leftmost network in \eqref{fig:diag-base-kron} depicts the
$k$-fold outer product of the network on the left
in \eqref{fig:diag-base} with itself. Observe that we simply take $k$ copies of the
network, but that for the purposes of the visualization we have taken care to draw the $k$ copies of each mode together for the socketing. In the middle in \eqref{fig:diag-base-kron},
we have adjoined two socket vertices to the input sockets to obtain
a realization $D^{\otimes k}$ of $A^{\otimes k}$. On the right in \eqref{fig:diag-base-kron},
we display an execution tree $\tdk$ for $D^{\otimes k}$. Observe how each
of the internal vertices of type (ii) in $\td$ gets expanded to a sequence
of $k$ internal vertices in $\tdk$. This transformation from $\td$ to $\tdk$
is the gist of the following theorem.

\begin{Thm}[Formal statement of \cref{thm:general upper bound}]
\label{thm:submul-am}
Let $D$ be an arbitrary realization of $A$ and let $\td$ be an arbitrary
execution tree for $D$. For all positive integers $k$, we have
\begin{equation}
\label{eq:submul-am}
c(A^{\otimes k})\leq a(\td)^{k-1} c(\td)\,.
\end{equation}
Furthermore, this realization of $A^{\otimes k}$ consists of
at most $k \cdot |V(D)|$ vertices.
\end{Thm}

\begin{Proof}
Let $D^*$ be the subnetwork of $D$ with $T(D^*)=\hat{T}(A)$. That is,
$D^*$ is the network induced by all the non-socket vertices of $D$.
Taking $k$ disjoint copies of $D^*$, we obtain a network whose tensor is $\hat{T}(A^{\otimes k})$. Attaching the resulting network
to tensors at sockets gives a realization of $A^{\otimes k}$.
Let us write $D^{\otimes k}$ for this realization.

To establish \eqref{eq:submul-am}, it suffices to construct an execution
tree $\tdk$ for $D^{\otimes k}$ whose cost satisfies
$c(\tdk)\leq a(\td)^{k-1} c(\td)$. We construct $\tdk$
by rewriting $\td$ from leaves towards the root
to consider the $k$ copies of each vertex in $D^*$. We start with leaf vertices which are the vertices of $D^{\otimes k}$.
We split the process into cases (i), (ii), and (iii)
as in the definition of amortized cost.
Let $x$ be the internal vertex of $\td$ that we are currently considering.

In case (i), we perform the contraction indicated by $x$ in each of the
$k$ copies of $D^*$ in $D^{\otimes k}$ individually. This creates $k$ new internal vertices
in $\tdk$ that are all copies of $x$.  We set these $k$ vertices as the vertices
that correspond to $x$ in the subsequent steps. Each of these contractions in $\tdk$ has the
same cost as the contraction indicated by $x$ in $\td$. This cost is less or equal than $c(\td)$.

In case (ii), let $y$ be the child of $x$ in $\td$ such that the subtree rooted at $y$ contains a socket vertex, and
let $z$ be the other child of $x$ in $\td$.
There is a single vertex in
$\tdk$ corresponding to $y$ and $k$ identical vertices
in $\tdk$
corresponding to $z$. We contract these $k$ vertices individually each with
the vertex that corresponds to $y$. This creates $k$ new internal vertices
in $\tdk$, where we set the topmost vertex as the vertex that corresponds
to $x$ in the subsequent steps. After the $i$th step, the corresponding tensor has $i$ copies of modes of $x$ and $k-i$ copies of modes of $y$. The cost of the  contraction in the $i$th step is the cost of contracting $y$ and $z$ in $\td$ multiplied by the the volume of $y$ to the power $k-i$ and the volume of $x$ to the power $i-1$. Since the volumes of $x$ and $y$ are less than or equal to $a(\td)$, this cost is less than or equal to $a(\td)^{k-1}c(\td)$.

In case (iii), let $y$ and $z$ be the two child vertices of $x$ in $\td$.
By the structure of the earlier steps, we have that a single vertex in
$D^{\otimes k}$ corresponds to $y$,
and similarly for $z$. We contract these two vertices. This creates one new
internal vertex in $\tdk$, which we set as the vertex that corresponds
to $x$ in the subsequent steps. This tensor has $k$ copies of modes of $x$. The cost of this contraction in $\tdk$ is the cost of the corresponding contraction in $\td$ to the $k^\text{th}$ power, because both tensors have $k$ copies of all modes compared to $y$ and $z$. By definition, in case (iii) the amortized cost of contracting $y$ and $z$ is the same as the cost of contracting $y$ and $z$. Hence the cost of this contraction in $\tdk$ is less or equal than $a(\td)^k \leq a(\td)^{k-1} c(\td)$.

This rewriting process produces an execution tree $\tdk$
for $D^{\otimes k}$ with $c(\tdk)\leq a(\td)^{k-1} c(\td)$.
\end{Proof}

An immediate corollary is that tensor networks can use low rank
decompositions of $T(A)$ to efficiently evaluate $A^{\otimes k}$.

\begin{Cor}[Submultiplicativity of low-rank executions]
  \label{cor:low rank execution}
  Let $A: \F^{J(E_1)} \times \F^{J(E_2)} \times \cdots \times
  \F^{J(E_\ell)} \rightarrow \F^{J(E')}$ be a multilinear map.  Define
  % Ugly space-fiddling hack to fix overfull hbox, remove "\!" to undo
  $m \!=\! \max\{|J(E_1)|, |J(E_2)|, \ldots, |J(E_\ell)|, |J(E')|\}$ and
  $r = \rk T(A)$.  Then $c(A^{\otimes k}) \le \max(r, m)^{k} \min(r, m)$
\end{Cor}

\begin{Proof}
  By taking a star-like network topology (as in \eqref{fig:diag-base})
  we get an execution with $a(\td) = \max(r, m)$ and cost $c(\td) = m \cdot
  r$.
\end{Proof}

\subsection{Fast matrix multiplication}
\label{sec:matrix mul upper bound}

Let us now illustrate the use of Theorem~\ref{thm:submul-am} by capturing
fast matrix multiplication with tensor networks. We start by considering square
matrix multiplication and then proceed to rectangular matrix multiplication.

We recall that the matrix multiplication exponent $\omega$
(cf.~\S\ref{sect:our}) satisfies \cite[Proposition~15.1]{Buergisser1997}
\begin{equation}
\label{eq:omega-inf}
\omega = \inf \!\left\{ \frac{\log \rk\,\bra n, n, n \ket_3}{\log n} \,:\, n \ge 2 \right\}\,.
\end{equation}
Above in \eqref{eq:omega-inf} we write $\bra n, n, n \ket_3$ for the
3-tensor obtained from the 6-tensor $\bra n, n, n \ket$ by flattening
the two modes corresponding the rows and columns of each of the three matrices.
That is, $\bra n, n, n \ket_3$ has shape $n^2\times n^2\times n^2$.

\begin{Lem}
\label{lem:mm}
For all $\epsilon>0$ it holds that two $n\times n$ matrices may be
multiplied in $O(n^{\omega+\epsilon})$ operations by executing a
tensor network.
\end{Lem}
\begin{Proof}
Fix an $\epsilon>0$. By \eqref{eq:omega-inf} we can
let $\alpha,\beta,\gamma$ be three $3$-tensors of
shape $(c\times c)\times d$ for positive integer constants $c\geq 2$ and $d$
with $c^2\leq d\leq c^{\omega+\epsilon/2}$
such that $\alpha,\beta,\gamma$ decompose the matrix multiplication
tensor $\bra c,c,c\ket$ defined by \eqref{eq:mm-t} as depicted below; the
indices $i_1,k_1,k_2,j_1,i_2,j_2$ refer to the tensor \eqref{eq:mm-t}.
The mode shared by $\alpha,\beta,\gamma$ in \eqref{fig:ccc} has length $d$,
all the other modes have length $c$.
\slimfigure{1.0}{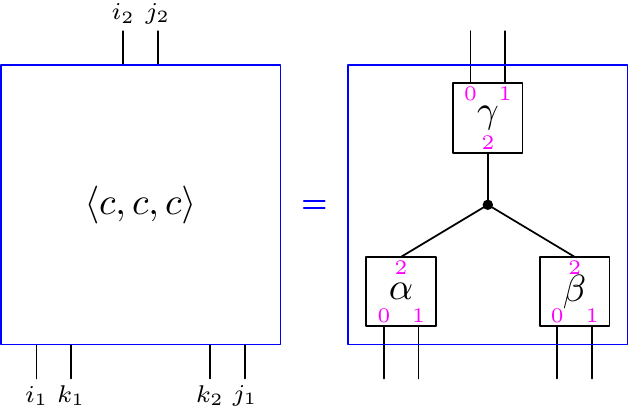}{ccc}
For example, Strassen's decomposition~\cite{Strassen1969} as depicted
in \eqref{fig:strassen-2-2-7} below realizes \eqref{fig:ccc} with $c=2$ and $d=7$.
We use the numbering in magenta to indicate correspondence between
modes in \eqref{fig:ccc} and \eqref{fig:strassen-2-2-7}.
\slimfigure{1.0}{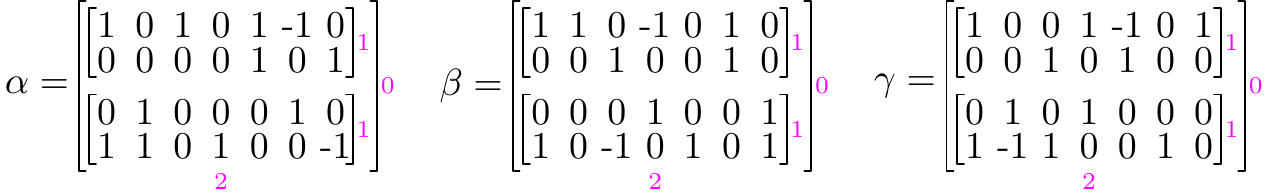}{strassen-2-2-7}

Let us next set up an application of Theorem~\ref{thm:submul-am}.
Let
$A:\F^{[c]\times [c]}\times \F^{[c]\times [c]}\rightarrow \F^{[c]\times [c]}$
be the bilinear multiplication map for two $c\times c$ matrices.
Observe that the tensor of $A$ is $\hat{T}(A)=\bra c,c,c\ket$.
To realize $A$, define two input sockets in \eqref{fig:ccc},
namely $\{i_1,k_1\}$ and $\{k_2,j_1\}$
to obtain a realization $D$ and an execution tree $\td$ as follows:
\slimfigure{1.0}{figures/ub_tensor_rank}{diag-base-2}
Since $c^2\leq d$, the amortized cost of $\td$ satisfies
$a(\td)=d$. The cost is $c(\td)=c^2d$.

Let the matrices $X,Y\in\F^{[n]\times[n]}$ be given. We construct
a tensor network that multiplies $X$ and $Y$.
We may assume that $n\geq c$. Let $k$ be the unique integer with
$c^{k-1}<n\leq c^k$. Extend the matrices $X,Y$ to
$X,Y\in\F^{[c^k]\times [c^k]}$ by inserting rows and columns with
zero-entries.

Since $\bra c,c,c\ket^{\otimes k}=\bra c^k,c^k,c^k\ket$ by
\eqref{eq:mm-kron}, we have that $A^{\otimes k}$ is the multiplication
map for two $c^k\times c^k$ matrices.
Using Theorem~\ref{thm:submul-am} with $D$ and $\td$, we obtain
$c(A^{\otimes k})\leq a(\td)^{k-1}c(\td)=d^{k}c^2$.
Moreover, the realization $D^{\otimes k}$ of $A^{\otimes k}$
given by Theorem~\ref{thm:submul-am} consists of
$|V(D^{\otimes k})|=O(k)$ vertices.
We can now associate $X$ and $Y$ with the
two socket vertices of $D^{\otimes k}$, taking care to
associate $X$ with the left socket (originating from $\{i_1,k_1\}$ and $\alpha$)
and $Y$ with the right socket (originating from $\{k_1,j_1\}$ and $\beta$).
(Cf.~\eqref{fig:diag-base-kron} for an illustration how $D$ and $\td$
in \eqref{fig:diag-base-2} yield $D^{\otimes k}$ and $\tdk$.)
Executing $D^{\otimes k}$ then results in the product matrix
$A^{\otimes k}(X,Y)=XY$ in
$d^{k}c^2|V(D^{\otimes k})|\leq
 c^{k(\omega+\epsilon/2)+2}|V(D^{\otimes k})|\leq
 c^{k(\omega+\epsilon)}=O(n^{\omega+\epsilon})$
operations for all large enough $n$.
\end{Proof}

Let us next proceed to rectangular matrix multiplication, where our
strategy is to reduce to square matrix multiplication. Also observe that
in the case $\omega=2$ the upper bounds in the following theorem
are optimal up to the choice of $\epsilon>0$ because of the size of
the input/output.

\begin{Lem}
\label{lem:rmm}
For all $\epsilon>0$ it holds that we may multiply an $n\times r$ matrix
with an $r\times n$ matrix by executing a tensor network in
\begin{enumerate}[(i)]
\item
$O(n^{\omega+\epsilon-1}r)$ operations when $r\geq n$, and
\item
$O(n^2r^{\omega+\epsilon-2})$ operations when $r\leq n$.
\end{enumerate}
\end{Lem}

\begin{Proof}
Fix an $\epsilon>0$ and let $\alpha,\beta,\gamma$ be three 3-tensors
of shape $(c\times c)\times d$ for constants $c$ and $d$
as in the proof of Lemma~\ref{lem:mm}.
Let the matrices $X\in\F^{[n]\times [r]}$ and $Y\in\F^{[r]\times[n]}$
be given. We construct tensor networks that compute the product $XY$.

To establish (i), first pad $X$ and $Y$ using rows and columns
of zero-entries so that both $n$ and $r$
become positive integer powers of $c$ and $n$ divides $r$.
This increases $n$ and $r$ by at most a multiplicative factor $c$.
We now have $n=c^k$ and $r=c^t$ for positive integers $t\geq k$.

Observe that we can compute the $n\times n$ product matrix $XY$ by
taking the sum of $\frac{r}{n}=c^{t-k}$ matrix products of size $n\times n$.

Let us implement this computation with a tensor network.
Reshape
$X$ to a $(k+t)$-tensor whose all modes have length $c$. The first $k$ modes
index the rows, the last $t$ modes index the columns.
Reshape
$Y$ to a $(t+k)$-tensor whose all modes have length $c$. The first $t$ modes
index the rows, the last $k$ modes index the columns.

Connect $X$ and $Y$ into a network as displayed in \eqref{fig:ccc-r-geq-n}
on the right.
\slimfigure{1.0}{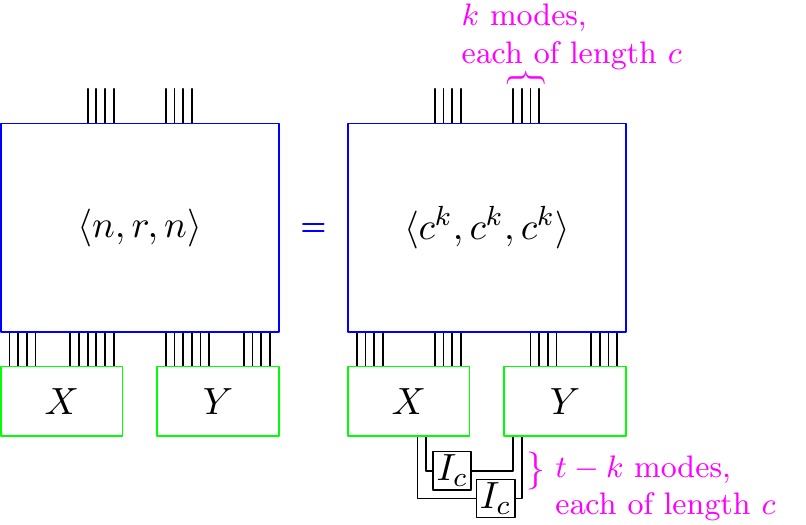}{ccc-r-geq-n}
That is,
we join $t-k$ column modes of $X$ with the matching $t-k$ row modes of $Y$
using $t-k$ identity matrices $I_c$
(to avoid degeneracy of the network if $X$ and $Y$ are removed).
Then we connect the remaining modes of $X$ and $Y$
to the left and right sockets of a matrix multiplication network
for $c^k\times c^k$ matrices as depicted by $\bra c^k,c^k,c^k\ket$
in \eqref{fig:ccc-r-geq-n}; this matrix multiplication network is obtained
as in the proof of Lemma~\ref{lem:mm}. The result is a multiplication
network $\bra n,r,n\ket$ as depicted in \eqref{fig:ccc-r-geq-n} on the left.
We execute this network using the structure on the right in
\eqref{fig:ccc-r-geq-n} by first
contracting (with zero cost)
the identity matrices $I_c$ with $Y$. We then execute the
remaining network as in the proof of Lemma~\ref{lem:mm}.
For large enough $n$, the cost of the execution is at most
$c^{k(\omega+\epsilon/2)+2}c^{t-k+1}$, which translates to
$O(n^{\omega+\epsilon}r)$ operations since the network
has $O(k)$ vertices whose contractions have nonzero cost.

To establish (ii), first pad $X$ and $Y$ using rows and columns
of zero-entries so that both $n$ and $r$
become positive integer powers of $c$ and $r$ divides $n$.
This increases $r$ and $n$ by at most a multiplicative factor $c$.
We now have $r=c^k$ and $n=c^t$ for positive integers $t\geq k$.

Observe that we can compute the $n\times n$ product matrix $XY$ by
taking $(\frac{n}{r})^2=c^{2(t-k)}$ matrix products of
size $r\times r$.

Let us implement this computation with a tensor network.
Reshape
$X$ to a $(t+k)$-tensor whose all modes have length $c$.
The first $t$ modes index the rows, the last $k$ modes index the columns.
Reshape
$Y$ to a $(k+t)$-tensor whose all modes have length $c$. The first $k$ modes
index the rows, the last $t$ modes index the columns.
Connect $X$ and $Y$ into a network as displayed in \eqref{fig:ccc-r-geq-n}
on the right.
\slimfigure{1.0}{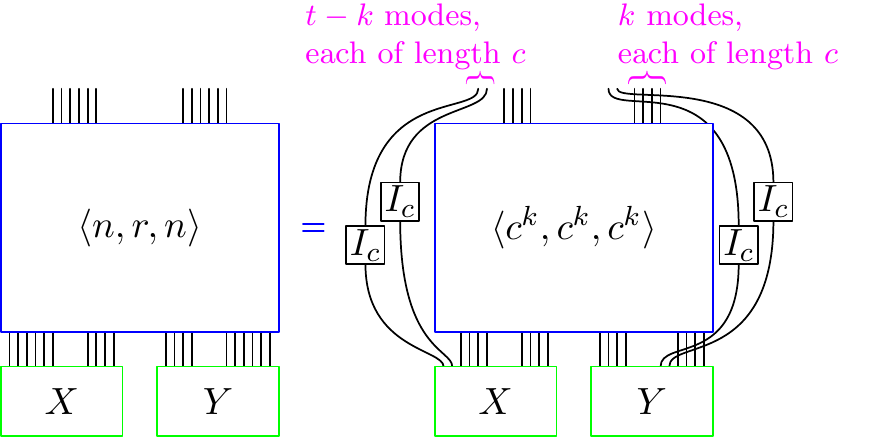}{ccc-r-leq-n}
That is, we join $t-k$ row modes of $X$ each to an identity matrix $I_c$
whose other mode is at the boundary. Similarly, we join matching $t-k$
column modes of $Y$ each to an identity matrix $I_c$ whose other mode
is at the boundary.
(This is to avoid degeneracy of the network if $X$ and $Y$ are removed.)
Then connect the remaining modes of $X$ and $Y$
to the left and right sockets of a matrix multiplication network
for $c^k\times c^k$ matrices as depicted by $\bra c^k,c^k,c^k\ket$
in \eqref{fig:ccc-r-leq-n}; this matrix multiplication network is obtained
as in the proof of Lemma~\ref{lem:mm}. The result is a multiplication
network $\bra n,r,n\ket$ as depicted in \eqref{fig:ccc-r-leq-n} on the left.
We execute this network using the structure on the right in
\eqref{fig:ccc-r-leq-n} by first
contracting (with zero cost)
the identity matrices $I_c$ with $X$ and $Y$, respectively.
We then execute the remaining network as in the proof
of Lemma~\ref{lem:mm}.
For large enough $r$, the cost of the execution is at most
$c^{k(\omega+\epsilon/2)+2}c^{2(t-k)+1}$, which translates to
$O(n^2r^{\omega-1+\epsilon})$ operations since the network
has $O(k)$ vertices whose contractions have nonzero cost.
\end{Proof}

Let us conclude this subsection with a well-known lemma on rectangular
matrix multiplication that we can also capture with tensor networks.

\begin{Lem}
  \label{lem:general rectangular mul}
  For all non-negative integers $a,b,c$ and $\epsilon > 0$
  it holds that we we may multiply an $n^a \times n^b$ matrix by an
  $n^b\times n^c$ matrix by a tensor network in
  $O(n^{\max(a+b,b+c,a+c) (\omega + \epsilon) / 2})$ operations.

\end{Lem}

\begin{proof}
Let $\epsilon>0$ be given.
By symmetry we may assume that $a\leq c$. Thus there are three cases
to consider, namely (i) $a\leq b\leq c$, (ii) $a\leq c\leq b$, and
(iii) $b\leq a\leq c$.

When $a\leq b\leq c$, we need to achieve $O(n^{(b+c)(\omega+\epsilon)/2})$
operations. Toward this end, it suffices to multiply an $n^b\times n^b$ matrix
with an $n^b\times n^c$ matrix, which can be implemented as $n^{c-b}$
multiplications of two square matrices of size $n^b\times n^b$.
Proceeding analogously as in \cref{lem:rmm}, we obtain a network that
can be executed in $O(n^{c-b}n^{b(\omega+\epsilon/2)})$ operations.
We have $c-b+b(\omega+\epsilon)\leq (b+c)(\omega+\epsilon)/2$ as desired
since $c\geq b$, $\omega\geq 2$, and $\epsilon\geq 0$.

When $a\leq c\leq b$, we need to achieve $O(n^{(b+c)(\omega+\epsilon)/2})$
operations. Toward this end, it suffices to multiply an $n^c\times n^b$ matrix
with an $n^b\times n^c$ matrix, and apply part (i) of \cref{lem:rmm} to obtain
a network that can be executed in $O(n^{b+c(\omega+\epsilon-1)})$ operations.
We have $b+c(\omega+\epsilon-1)\leq (b+c)(\omega+\epsilon)/2$ as desired
since $b\geq c$, $\omega\geq 2$, and $\epsilon\geq 0$.

When $b\leq a\leq c$, we need to achieve $O(n^{(a+c)(\omega+\epsilon)/2})$
operations. Toward this end, it suffices to multiply an $n^a\times n^b$ matrix
with an $n^b\times n^c$ matrix. An easy modification of part (ii) of
\cref{lem:rmm} gives a network that can be executed in
$O(n^{a+c+b(\omega+\epsilon-2)})$ operations.
We have $a+c+b(\omega+\epsilon-2)\leq (a+c)(\omega+\epsilon)/2$ as desired
since $b\leq a$, $b\leq c$, $\omega\geq 2$, and $\epsilon\geq 0$.
\end{proof}

\subsection{Homomorphism-counting for pattern graphs of small branchwidth}

The following upper bound for $P$-linear forms when $P$ is a graph
of small branchwidth is our main result in this section.

\begin{Lem}
  \label{lem:branchwidth ub}
  For any fixed pattern graph $P$ and every $\epsilon > 0$, there is a tensor network that evaluates the $P$-linear form of order $n$
  in $O(n^2 + n^{\bw(P) (\omega + \epsilon) / 2})$ operations.
\end{Lem}

\begin{proof}
  Let $\bw(P) = w$ and consider any branch decomposition of $P$ of
  width $w$.  Rooting this decomposition arbitrarily by subdividing
  any edge and taking the newly added vertex as root, we obtain a
  binary rooted tree $T$ with $|E(P)|$ leaves where the leaves are
  identified by the edges of $P$.  Let $r$ be the root of $T$ and
  for each vertex $u$ of $T$, let:
  \begin{enumerate}
  \item $C_u \subseteq V(P)$ (mnemonic: $C$ for ``crossing'' or ``cut'') be the set of
    all vertices of $P$ that appear both in some leaf in the subtree
    rooted at $u$, and in some leaf outside the subtree.  By
    definition $|C_u| \le w$ for all $u$, and $C_r = \emptyset$.

  \item $D_u \subseteq V(P)$ (mnemonic: $D$ for ``done'') be the set of all vertices of $P$ that
    appear only in leaves in the subtree rooted at $u$, and not
    outside.  Note that $C_u$ and $D_u$ form a partition of the set of
    all vertices appearing in some leaf in the subtree rooted at $u$,
    and that $D_r = V(P)$.

  \item $E_u \subseteq E(P)$ be the set of all leaves of the subtree
    of $T$ rooted at $u$ (recall that each leaf of $T$ corresponds to
    an edge).  Note that $E_r = E(P)$.

  \item $A_u \in \F^{[n]^{C_u}}$ be the following order $|C_u|$ tensor
    of shape $n \times n \times \cdots \times n$ defined for a position $i \in [n]^{C_u}$ by
    \[
    (A_u)_i = \sum_{j \in [n]^{D_u}} \prod_{S \in E_u} X^{(S)}_{(ij)_{|S}},
    \]
    where $ij \in [n]^{C_u \cup D_u}$ is the Cartesian product of $i$
    and $j$.
  \end{enumerate}

  Note that $A_r$ equals the value of the $P$-linear
  form \eqref{eq:homomorphism_form}.
  Furthermore, for a leaf $u$ of $T$ corresponding to an edge $\{x_1, x_2\}$ of
  $P$, $A_u$ is easy to compute: if both $x_1$ and $x_2$ have degree
  at least $2$, then $A_u$ simply equals the input tensor
  $X^{\{x_1,x_2\}}$.  On the other hand if $x_1$ and/or $x_2$ has
  degree $1$, then $A_u$ is either a vector or a scalar, being either
  the row sums, column sums, or sum of all entries of
  $X^{\{x_1,x_2\}}$.  Each of these is readily computed in $O(n^2)$
  time by a contraction of $X^{\{x_1,x_2\}}$ with an appropriate
  tensor of ones.

  As we shall see below, for every non-leaf vertex $x$ of $T$ with
  children $y$ and $z$, the tensor $A_x$ equals the contraction of
  $A_y$ and $A_z$.
  Thus, the desired result would follow if we can show that the
  contraction of two siblings $A_y$ and $A_z$ can be computed by a
  tensor network in $O(n^{w (\omega+\epsilon)}/2)$ operations.  We
  now proceed to establish this.

    Partition $C_x$ into $C_{xy} \cup C_{xz} \cup C_{xyz}$ where
    $C_{xyz} = C_x \cap C_y \cap C_z$, $C_{xy} = C_x \setminus C_z$
    and $C_{xz} = C_x \setminus C_y$ (note that every element of $C_x$
    must appear in exactly one of $C_y$ and $C_z$ so these three sets
    indeed partition $C_x$).  Symmetrically partition $C_y$ into
    $C_{xy} \cup C_{yz} \cup C_{xyz}$ and $C_z$ into
    $C_{xz} \cup C_{yz} \cup C_{xyz}$.  Note that $C_{yz} \subseteq D_x$ and that the contraction of $A_y$ and $A_z$ at position $(i,j,k)$ for $i \in [n]^{C_{xy}}$, $j \in [n]^{C_{xz}}$, $k \in [n]^{C_{xyz}}$ is exactly
    \[
    \sum_{\ell \in [n]^{C_{yz}}} A_y(i,k,\ell) A_z(j,k,\ell) =  A_x(i,j,k)
    \]
    Let $a = |C_{xy}|$, $b = |C_{yz}|$, and $c = |C_{xz}|$.  Split
    each mode in $C_{xy}$ into two separate modes $C_{xy}^{(x)}$ and
    $C_{xy}^{(y)}$ where the former is used to index $A_x$ and the
    latter is used to index $A_y$.  Similarly split the modes in $C_{xz}$ and
    $C_{yz}$, but \emph{not} the modes in $C_{xyz}$.

    The contraction of $A_y$ and $A_z$ can then be evaluated using
    rectangular matrix multiplication with the following tensor
    network.
    \slimfigure{1.0}{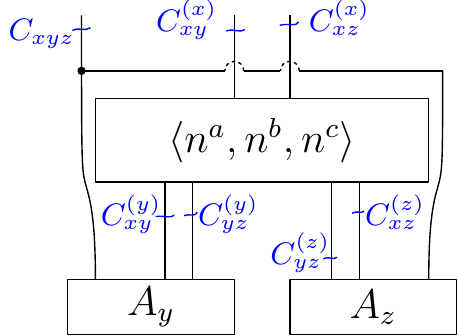}{branchwidth}

    By \cref{lem:general rectangular mul}, this network can be
    evaluated in
    $O(n^{|C_{xyz}|} n^{\max(a+b, b+c, a+c) (\omega+\epsilon)/2})$
    operations.  Since $a + b + |C_{xyz}| = |C_y| \le w$,
    $b+c + |C_{xyz}| = |C_z| \le w$, and
    $a+c + |C_{xyz}| = |C_x| \le w$, the number of operations used to
    contract $A_y$ and $A_z$ is bounded by
    $O(n^{w (\omega+\epsilon)/2})$.
\end{proof}

\subsection{Fast Fourier transforms and fast convolution}

\label{sect:fft}

Let us next express fast Fourier transforms as tensor networks and their
executions.
We start by observing that the Cooley--Tukey~\cite{CooleyT1965} fast
Fourier transform on $\Z_{2^k}$ can be implemented as a tensor network.
We assume $\rho$ is a primitive $(2^k)^\text{th}$ root of unity in
the field $\F$.

\begin{Lem}
\label{lem:fft}
The discrete Fourier transform for the Abelian group $\Z_{2^k}$ can
be computed by executing a tensor network in $O(2^kk)$ operations.
\end{Lem}
\begin{Proof}
The case $k=0$ is immediate so let us assume that $k\geq 1$.
We construct a tensor
network whose execution multiplies a vector $x\in\F^{[2^k]}$ with the DFT
matrix $\Phi$ in \cref{sect:dft} for $n=2^k$
to yield the result $\Phi x\in\F^{[2^k]}$.
Toward this end, let us write $I_m$ for an
$m\times m$ identity matrix,
\begin{equation}
\label{eq:h2}
H_2=\left[\begin{array}{rr}1&1\\1&-1\end{array}\right]
\end{equation}
for the $2\times 2$ Hadamard-Walsh matrix, and
$R^{(k,j)}\in\F^{[2^k]}$ for the vector obtained by concatenating
$2^j$ copies of the vector
\[
[\underbrace{1,1,\ldots,1}_{2^{k-1-j}},
 \underbrace{\rho^{2^j\cdot 0},\rho^{2^j\cdot 1},\ldots,\rho^{2^j\cdot(2^{k-1-j}-1)}}_{2^{k-1-j}}]
\]
for $j=0,1,\ldots,k-2$.
We can (e.g.~\cite{VanLoan1992}) decompose $\Phi$ into a sequence
of $2^k\times 2^k$ matrices
\begin{equation}
\label{eq:cooley-tukey}
\Phi=P^{(k)}B^{(k,k-1)}T^{(k,k-2)}B^{(k,k-2)}\cdots T^{(k,1)}B^{(k,1)}T^{(k,0)}B^{(k,0)}
\end{equation}
consisting of alternating butterfly matrices
\[
B^{(k,j)}=
\underbrace{I_2\otimes I_2\otimes \cdots\otimes I_2}_{j}
\otimes H_2 \otimes
\underbrace{I_2\otimes I_2\otimes \cdots\otimes I_2}_{k-j-1}
\]
and diagonal twiddle matrices
\[
T^{(k,j)}=\text{diag}(R^{(k,j)})
\]
followed by a permutation matrix $P^{(k)}$ that permutes the indices
in $[2^k]$ viewed as $k$-bit strings by reversing the bit-order.
Since only the $(j+1)^\text{th}$ Kronecker component in the butterfly
matrix $B^{(k,j)}$ is a nonidentity matrix, and multiplication of
a vector with the diagonal twiddle matrix $T^{(k,j)}$
corresponds to pointwise (Hadamard) multiplication with the
vector $R^{(k,j)}$ on the diagonal,
we observe that the sequence \eqref{eq:cooley-tukey}
can be represented as a tensor network $D^*$ as depicted below (for $k=5$)
so that all the modes have length $2$.
\slimfigure{0.62}{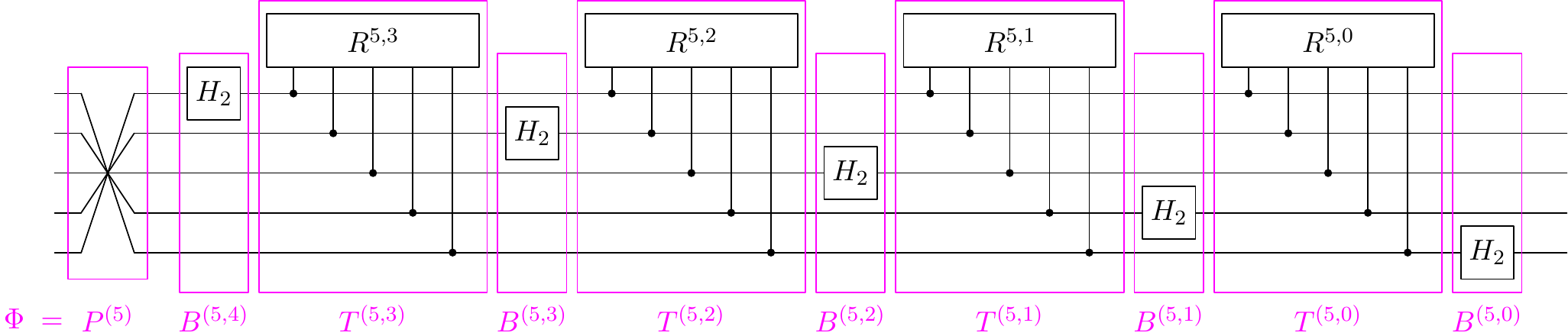}{diag-phi-5}
We can now connect the network \eqref{fig:diag-phi-5} to a data vector
$x\in\F^{[2^k]}$ to obtain the network $D$ below:
\slimfigure{0.62}{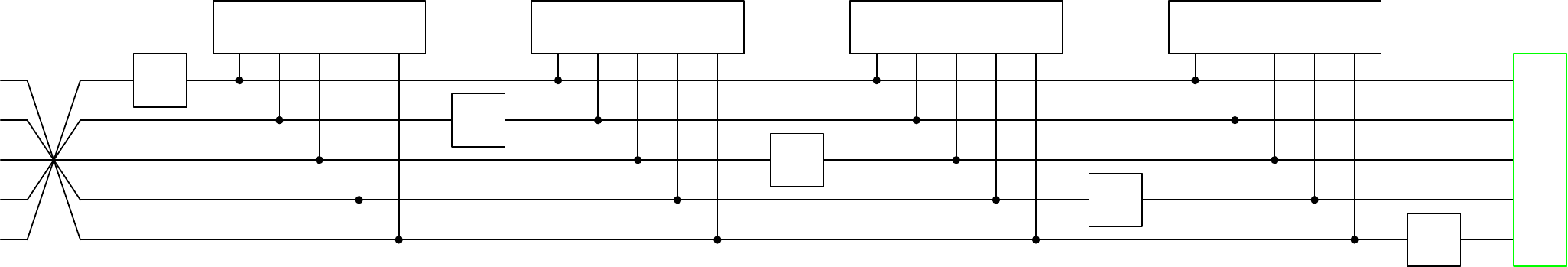}{diag-phi-x-5}
Below we depict in red an execution tree $\td$ with cost $2^{k+1}$
for the network \eqref{fig:diag-phi-x-5}. Observe that the mode
permutation (multiplication with the matrix $P^{(5)}$) is not part of
the execution since the permutation amounts to merely rearranging
the modes.
\slimfigure{0.62}{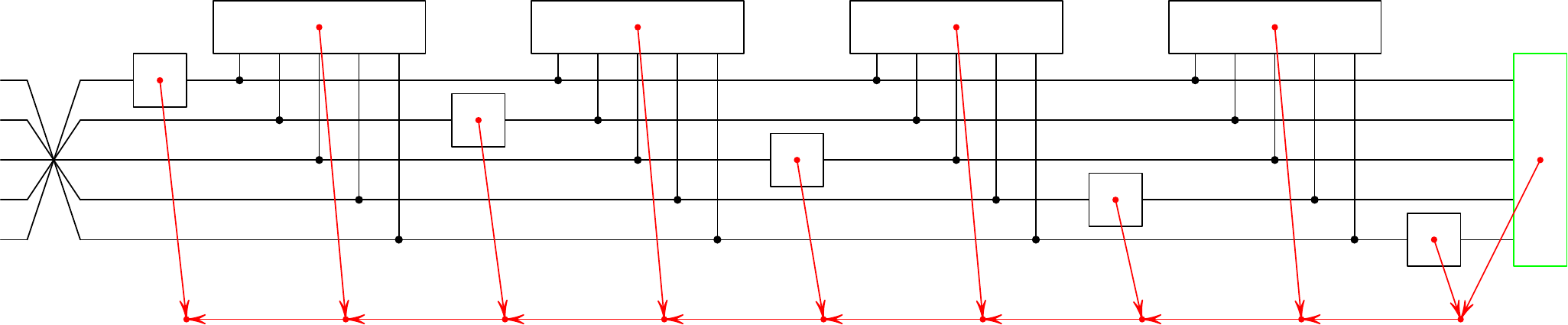}{diag-phi-x-5-exec}
Since the network has $O(k)$ tensors, the execution
\eqref{fig:diag-phi-x-5-exec} can be carried out in $O(2^kk)$ operations in $\F$.
\end{Proof}

\begin{Lem}
\label{lem:fwht}
The discrete Fourier transform for the elementary Abelian group $\Z_{2}^k$ can
be computed by executing a tensor network in $O(2^kk)$ operations.
\end{Lem}
\begin{Proof}
This proof is analogous to the proof of Lemma~\ref{lem:fft} but omits the twiddle matrices
$T^{(k,j)}$ and the final permutation matrix $P^{(k)}$ from the
decomposition of the tensor network.
\end{Proof}

The following corollary is immediate from Lemma~\ref{lem:fft}.

\begin{Lem}
\label{lem:conv}
The group algebra products on $\F[\Z_{2^k}]$ and $\F[\Z_2^k]$
can be computed by executing a tensor network in $O(2^kk)$ operations whenever
$2$ is a unit in $\F$.
\end{Lem}
\begin{Proof}
We start with $\F[\Z_{2^k}]$.
Let $f,g\in\F^{[2^k]}$ be two vectors given as input. Our task is to
compute the $\Z_{2^k}$-convolution $f*g\in\F^{[2^k]}$.
The case $k=0$ is immediate so suppose that $k\geq 1$.
Recalling that $f*g=\Phi^{-1}((\Phi f)\cdot (\Phi g))$, where
``$\cdot$'' denotes an elementwise (Hadamard) product of two vectors of
length $2^k$, let us construct a tensor network as follows.
First take the FFT of $f$ and $g$ using Lemma~\ref{lem:fft},
then multiply the resulting vectors elementwise, and
finally take the inverse FFT by
replacing $\rho$ with $\rho^{-1}$ in Lemma~\ref{lem:fft} and
multiplying with the diagonal matrix $\frac{1}{2^k}I_{2^k}$.
Below we display (for $k=5$) the resulting network $D$ that executes to $f*g$.
\slimfigure{0.28}{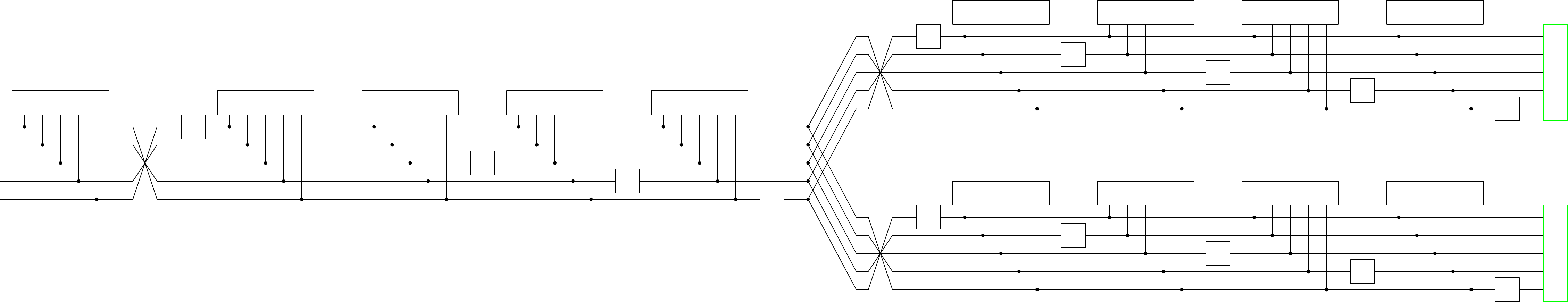}{diag-conv}
The execution of the network proceeds from right to left analogously to
\eqref{fig:diag-phi-x-5-exec}.
\slimfigure{0.28}{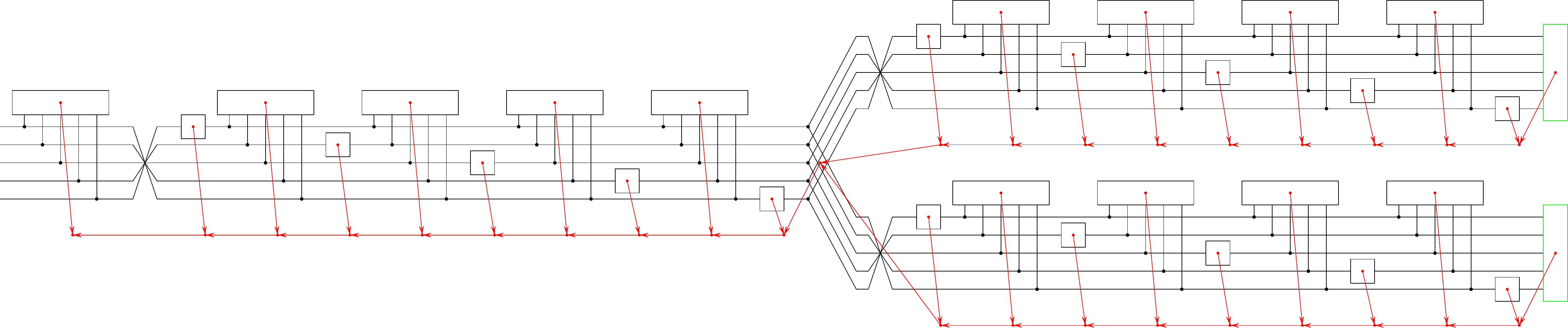}{diag-conv-exec}
The cost of this execution is $2^{k+1}$.
Since the network has $O(k)$ tensors, the execution
\eqref{fig:diag-conv-exec} can be carried out in $O(2^kk)$ operations in $\F$.

The case of $\F[\Z_2^k]$ is analogous but replacing Lemma~\ref{lem:fft}
with Lemma~\ref{lem:fwht} and modifying the diagrams accordingly.
\end{Proof}

\subsection{Yates' algorithm}

A particularly simple use case for Theorem~\ref{thm:submul-am}
occurs when $A:\F^{[s]}\rightarrow\F^{[t]}$ is a linear map.
It is immediate that we can realize $A$ with a two-vertex
network and an execution tree that has amortized cost $\max(s,t)$
and cost $st$.

Then, Theorem~\ref{thm:submul-am} immediately implies that we can evaluate
$A^{\otimes k}:\F^{[s]^k}\rightarrow\F^{[t]}$ using a tensor network
with cost $\max(s^{k+2},t^{k+2})$ and $O(k)$ vertices. In particular, the
network can be executed in $O(\max(s^{k+2},t^{k+2})k)$ operations in $\F$.
This network in essence realizes Yates' algorithm~\cite{Yates1937}
for multiplying an $s^k$-length vector with the $k^\text{th}$
Kronecker power of an $s\times t$ matrix to obtain $t^k$-length vector.

Applying the previous observation to $A=H_2$ in \eqref{eq:h2} with $s=t=2$,
we obtain Lemma~\ref{lem:fwht} as an immediate corollary. Similarly,
other choices of $2\times 2$ matrices yield the algebraic core of currently
the fastest known algorithms for problems such as graph coloring
and its generalizations such as computing the Tutte polynomial of
a graph~\cite{BjorklundHKK2007,BjorklundHKK2008,BjorklundHK2009}.
In particular, the pair of mutually inverse $2\times 2$ matrices
\[
Z_2=\left[\begin{array}{rr}
1&1\\
0&1
\end{array}\right]\,,\qquad
M_2=\left[\begin{array}{rr}
1&-1\\
0&1
\end{array}\right]
\]
yield, as $Z_2^{\otimes k}$ and $M_2^{\otimes k}$, the zeta and
M\"obius tranforms for the lattice $(\{0,1\}^k,\subseteq,\cap,\cup)$
of all subsets of a $k$-element set, partially ordered by subset inclusion.
Theorem~\ref{thm:submul-am} yields immediately the standard algorithms
(normally developed via Yates's algorithm)
for the fast zeta and the fast M\"obius transforms via tensor networks.
These networks can then be combined as in the proof of
Lemma~\ref{lem:conv} to yield the associated bilinear convolution maps
(multiplication maps in the semigroup algebra
$\F[(\{0,1\}^k,\subseteq,\cap,\cup)]$) to realize these maps in
$O(2^kk)$ operations. We omit the details due to similarity
with Lemma~\ref{lem:conv}.

\subsection{Kruskal operator}

We proceed to implement the Kruskal operator by Kroneckering and reduction
to fast rectangular matrix multiplication.

\begin{Lem}
For all constants $\epsilon>0$ and $\ell=1,2,\ldots$
it holds that we may evaluate the Kruskal operator
of $\ell$ matrices of shape $n\times r$ by executing a tensor network in
\begin{enumerate}[(i)]
\item
$O(n^{\lceil\ell/2\rceil(\omega+\epsilon-1)}r)$ operations when $r\geq n^{\lceil \ell/2\rceil}$, and
\item
$O(n^{2\lceil\ell/2\rceil}r^{\omega+\epsilon-2})$ operations when $r\leq n^{\lceil \ell/2\rceil}$.
\end{enumerate}
\end{Lem}
\begin{Proof}
Fix an $\epsilon>0$ and let $\alpha,\beta,\gamma$ be three 3-tensors
of shape $(c\times c)\times d$ for constants $c$ and $d$
as in the proof of Lemma~\ref{lem:mm}.
Let $A^{(1)},A^{(2)},\ldots,A^{(\ell)}\in\F^{[n]\times[r]}$ be given as input.
We construct a tensor network that computes the output $Y$ of the
Kruskal operator \eqref{eq:kruskal}.

Without loss of generality we may assume that $\ell$ is even by
introducing a matrix $A^{(\ell+1)}\in\F^{[n]\times[r]}$ filled with
$1$-entries and setting $\ell\leftarrow\ell+1$. By inserting
rows and columns with zero-entries as necessary,
we can assume that both $n$ and $r$ are positive integer powers of $c$.
The key idea is now to Kronecker the matrices
$A^{(1)},A^{(2)},\ldots,A^{(\ell/2)}$ and
$A^{(\ell/2+1)},A^{(\ell/2+2)},\ldots,A^{(\ell)}$
{\em in the vertical dimension only} to obtain two matrices
$B\in\F^{[n^{\ell/2}]\times [r]}$
and
$C\in\F^{[n^{\ell/2}]\times [r]}$,
respectively. We then multiply $B$ and the transpose of $C$ using fast
rectangular matrix multiplication from Lemma~\ref{lem:rmm} to obtain
$Y\in\F^{[n^{\ell/2}]\times [n^{\ell/2}]}$. It is immediate from
Lemma~\ref{lem:rmm} that this
results in the operation counts claimed in (i) and (ii) as long as we
can realize the idea using tensor networks.
Kroneckering in the vertical dimension only can be realized by joining
all the {\em horizontal} dimensions to a common mode, which becomes the
inner mode for matrix multiplication. The resulting network is
depicted below (for $\ell=6$), where either \eqref{fig:ccc-r-geq-n}
or \eqref{fig:ccc-r-leq-n} is used for the
subnetwork indicated with $\bra n^{\ell/2},r,n^{\ell/2}\ket$
depending on whether (i) or (ii) holds.
\slimfigure{1.0}{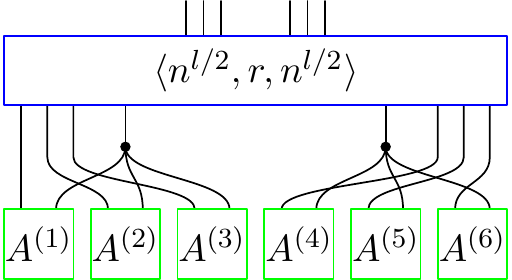}{kruskal-diag}
In drawing \eqref{fig:kruskal-diag} we have made two abstractions. First,
each drawn mode in fact is a bundle of modes, each of length $c$.
Second, each mode in a bundle that is incident to one of the input matrices
$A^{(1)},A^{(2)},\ldots,A^{(\ell)}$ is in fact subdivided by inserting
an identity matrix $I_c$ just before the incidence to the input matrix.
The network \eqref{fig:kruskal-diag} is executed first by contracting
the (zero-cost) identity matrices $I_c$, then contracting
$A^{(1)},A^{(2)},\ldots,A^{(\ell/2)}$ and
$A^{(\ell/2+1)},A^{(\ell/2+2)},\ldots,A^{(\ell)}$, and finally proceeding
to execute the subnetwork $\bra n^{\ell/2},r,n^{\ell/2}\ket$ as
in Lemma~\ref{lem:rmm}.
\end{Proof}

\subsection{Clique-counting forms}

For counting $v$-cliques, the general upper bound for $P$-linear forms
(\cref{lem:branchwidth ub}) with $P=K_v$ gives a running time of
$O(n^{(\omega+\epsilon)\lceil 2v/3 \rceil /2}) = O(n^{(\omega + \epsilon)\lfloor v/3
  \rfloor + \tfrac{(\omega+\epsilon)}{2}(v \bmod 3)})$.
In this section we give a slightly improved tensor-network algorithm,
matching the running time of Nešetřil and Poljak \cite{Nesetril1985}.
For brevity, throughout this section we refer to the $P$-linear form
for $P=K_v$ as the $\binom{v}{2}$-linear form.

\begin{Lem}
\label{lem:v choose 2 linear form upper bound}
For all constants $\epsilon>0$ and $v=3,4,\ldots$
it holds that we may evaluate the $\binom{v}{2}$-linear form of order $n$
by executing a tensor network in
$O(n^{(\omega+\epsilon)\lfloor v/3 \rfloor + (v \bmod 3)})$ operations.
\end{Lem}
\begin{Proof}
Fix an $\epsilon>0$ and let $\alpha,\beta,\gamma$ be three 3-tensors
of shape $(c\times c)\times d$ with constants $c$ and $d$
as in the proof of Lemma~\ref{lem:mm}.

Let $u=v\bmod 3$ and $s=(v-u)/3$. That is, $s$ and $u$ are positive
integers that satisfy $v=3s+u$. Let us write $[a,b]=[b]\setminus [a-1]$
and set
\[
\begin{split}
P_1&=[1,2s]\cup [3s+1,3s+u]\,,\\
P_2&=[1,s]\cup [2s+1,3s]\cup [3s+1,3s+u]\,,\\
P_3&=[s+1,3s]\,.
\end{split}
\]
Now partition the subsets in $\binom{[v]}{2}$ into three groups $G_1,G_2,G_3\subseteq\binom{[v]}{2}$ such that
(a) $G_1$ consists of all subsets $\{x,y\}\in\binom{[v]}{2}$
with $\{x,y\}\subseteq P_1$,
(b) $G_2$ consists of all subsets $\{x,y\}\in\binom{[v]}{2}\setminus G_1$
with $\{x,y\}\subseteq P_2$, and
(c) $G_3$ consists of all subsets $\{x,y\}\in\binom{[v]}{2}\setminus (G_1\cup G_2)$. In particular, we observe that for all
$\{x,y\}\in G_3$ we have $\{x,y\}\subseteq P_3$.
Also observe that $|P_1|=|P_2|=2s+u$ and $|P_3|=2s$.

Let the $\binom{v}{2}$ matrices $A^{\{x,y\}}\in\F^{[n]\times [n]}$ for
$\{x,y\}\in\binom{[v]}{2}$ be given as input. By inserting rows and
columns with zero-entries as appropriate, we may assume that $n=c^k$ for
a positive integer $k$. The key idea is to start with a base construction
for an input of size $c$ and then scale the construction up
to an input of size $n$ using Theorem~\ref{thm:submul-am}.

Let us now describe the base construction. First, we
introduce a subnetwork that forces consistency of modes for the inputs indexed
by $G_1$, $G_2$, and $G_3$, respectively.
In essence, for each mode
$z\in P_1$, introduce a mode for each occurrence of $z$ in $G_1$.
Similarly, for each $z\in P_2$, introduce a mode for each occurrence
of $z$ in $G_2$. Finally, for each $z\in P_3$, introduce a mode for
each occurrence of $z$ in $G_3$. Below we show an illustration for
$v=8$ and thus $s=2$, $u=2$.
\slimfigure{0.50}{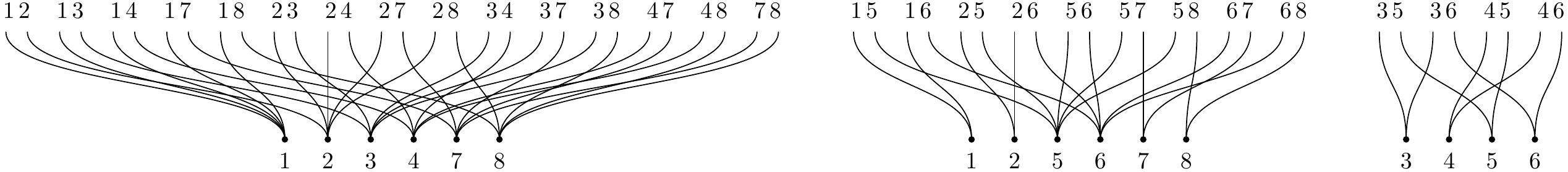}{v-2-subsume}
In the drawing \eqref{fig:v-2-subsume} all modes have length $c$, and
each occurrence of $z\in [v]$ at the top is tacitly assumed to be subdivided
by an identity matrix $I_c$. (Otherwise the network would be degenerate.)
Next we introduce $s$ copies of $\alpha,\beta,\gamma$ to force consistency
for the modes in $[1,3s]$ between $G_1,G_2,G_3$. We also join individually
each of the $u$ modes in $[3s+1,3s+u]$ between $G_1$ and $G_2$.
Below we show an illustration for $v=8$. Observe from \eqref{fig:ccc}
and the numbering in magenta that each introduced copy
$i=1,2,\ldots,s$ of $\alpha,\beta,\gamma$ forces consistency between
the modes $i,s+i,2s+i$.
\slimfigure{0.49}{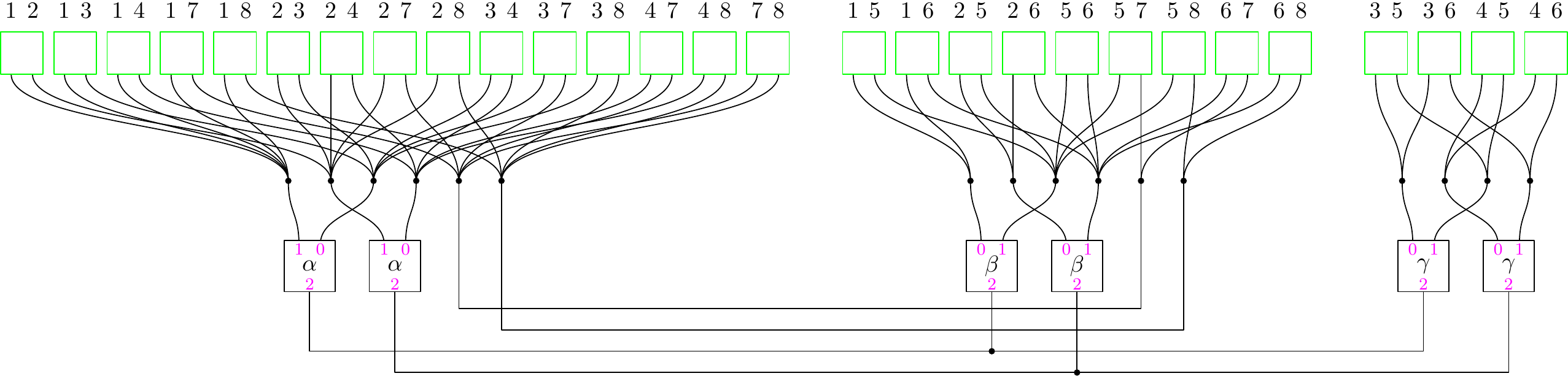}{v-2-abc}
The base network $D$ realizing the $c$-uniform $\binom{v}{2}$-linear form
is now as depicted in \eqref{fig:v-2-abc}.
We execute the network $D$ using the following execution tree $\td$,
where the identity matrices $I_c$ immediately incident to the
$\binom{v}{2}$ socket vertices are contracted tacitly and not drawn
in \eqref{fig:v-2-exec}.
\slimfigure{0.49}{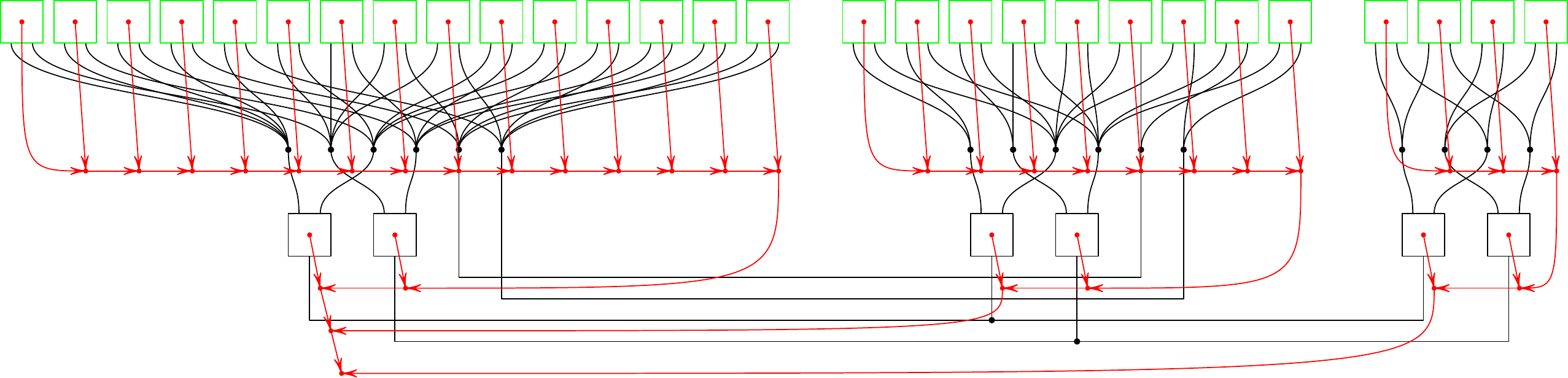}{v-2-exec}
The amortized cost of the execution in \eqref{fig:v-2-exec} is
$a(\td)=d^sc^u$ and the cost is $c(\td)\leq d^{s+1}c^{u+2}$.
By Theorem~\ref{thm:submul-am}, this translates to cost
at most $a(\td)^{k-1}c(\td)=d^{sk+1}c^{uk+2}$ for the
$n$-uniform $\binom{v}{2}$-linear form. Since the network given
by Theorem~\ref{thm:submul-am} has $O(sk)$ tensors
and $d\leq c^{\omega+\epsilon/2}$, we conclude
that the number of operations is $O(n^{(\omega+\epsilon)s+u})$.
The lemma follows.
\end{Proof}

\subsection{The permanent}

The following lemma observes that essentially the fastest known algorithm
for the permanent, namely Ryser's algorithm~\cite{Ryser1963}, can be
realized as a tensor network. Here by essentially fastest known we
mean the base of the exponential running time. Sub-exponential speed-ups
to Ryser's algorithm are known, see e.g.~Bj\"orklund~\cite{Bjorklund2016}.

\begin{Lem}
The permanent of an $n\times n$ matrix can be computed by executing
a tensor network in $O(2^nn)$ operations.
\end{Lem}
\begin{Proof}
We observe that Ryser's algorithm~\cite{Ryser1963} for the
permanent \eqref{eq:det and per}, namely the inclusion--exclusion expression
\[
\per A=\sum_{S\subseteq [n]}(-1)^{n-|S|}\prod_{i\in [n]}\sum_{j\in S}a_{ij}
\]
is implementable with a star-shaped tensor network consisting of
$n$ matrices of shape $2^n\times n$ joined together by a common mode (of length $2^n$), with the $n$ modes of length $n$ being the boundary of the network.
Each of the $n$ matrices consists of the $\{0,1\}$-valued incidence vectors of
the $2^n$ subsets $S\subseteq [n]$, with one of the matrices arbitrarily
selected to contain signed rows determined by $(-1)^{n-|S|}$.
The input to the network consists of $n$ vectors of length $n$,
namely the rows of the $n\times n$ input matrix $A$. The network is executed
by first executing the $n$ matrix-vector multiplications, and then contracting
the resulting $n$ vectors of length $2^n$ until the scalar $\per A$ remains.
\end{Proof}

\section{A lower bound for the cost of a multilinear map}
\label{sec:general lower bound}

In this section, we prove a general lower bound on the cost of
evaluating a multilinear map using tensor networks, as defined in
Section~\ref{sec:map cost}.  The lower bound is expressed in terms of the \emph{socket-width} of a multilinear map, which we
now proceed to define.

Let $A:\F^{J(E_1)}\times\F^{J(E_2)}\times\cdots \times\F^{J(E_\ell)}
\rightarrow \F^{J(E')}$ be an $\ell$-linear map.  A {\em socket-tree}
of $A$ is a tree $\ts$ whose $\ell+1$ leaf vertices are the sockets
$E_1,E_2,\ldots,E_\ell, E'$ of $A$ and whose internal vertices all have
degree exactly $3$.  Associate with each edge $e = \{x_R, x_C\}$ of $\ts$
the two subtrees $\ts(x_R, e)$ and $\ts(x_C, e)$ obtained by
removing $e$, where $\ts(x_R, e)$ is the subtree containing $x_R$ and $\ts(x_C,
e)$ is the subtree containing $x_C$.  Let $L(x_R, e)$
be the set of leaves in $\ts(x_R, e)$ and let $L(x_C, e)$
be the set of leaves in $\ts(x_C, e)$.

The sets $L(x_R, e)$ and $L(x_C,e)$ are both nonempty and together
partition the set of sockets.  Consider the flattening $M(\ts,e)$ of the tensor $T(A)$ such that the modes in
$L(x_R,e)$ index the rows and the modes in  $L(x_C,e)$ index
the columns of $M(\ts,e)$.  The \emph{width} of $\ts$ at $e$ is the
rank of $M(\ts, e)$, and the \emph{width} of $\ts$ is
$w(\ts)=\max_{e\in E(\ts)} \rk(M(\ts,e))$.

Let us write $\mathscr{S}(A)$ for the set of all socket-trees of the
multilinear form $A$.  We define the \emph{socket-width} of $A$ to be
$w(A)=\min_{\ts\in\mathscr{S}(A)} w(\ts)$.

The rest of this section is devoted to proving \cref{thm:general lower bound}:

\GeneralLowerBoundTheorem*

First, we prove that without loss of generality, we may restrict
attention to forms rather than general maps.

\begin{Claim}
For any multilinear map $A$, it holds that $c(A) \ge c(F(A))$.
\end{Claim}
\begin{Proof}
We observe that $A$ and $F(A)$ satisfy $\hat{T}(A)=\hat{T}(F(A))$. Any network
$D\in \mathscr{D}(A)$ can be modified to a network
$D'\in \mathscr{D}(F(A))$ by attaching a tensor $X' \in \F^{J(E')}$ to the boundary
of $D$. Let $D\in \mathscr{D}(A)$ be such that
$c(D)=c(A)$. The minimum-cost execution of $D$, followed
by contracting $T(D)$ and $X'$, is an execution of $D'$. Its
cost is $c(A)$, since the cost of contracting of $T(D)$ and $X'$ is $\prod_{e \in B(D)} |J(e)|$ and $\prod_{e \in B(D)} |J(e)| \leq c(A)$, because the last step of the minimum-cost execution of $D$ contracted a set $W$ with all modes $e \in B(D)$ incident to $W$. Thus, $c(A)\ge c(F(A))$.
\end{Proof}

Furthermore, $w(A) = w(F(A))$ for every multilinear map $A$, since
$w(A)$ only depends on the tensor $T(A)$, but not on which of its coordinates (if any) is the output.
Thus it suffices to prove \cref{thm:general lower bound} for
multilinear forms, which we now proceed to do.

\begin{Lem}
  For any multilinear form $F$, it holds that $c(F) \ge w(F)$.
\end{Lem}

\begin{Proof}
  Let $D\in\mathscr{D}(F)$ be such that $c(D)=c(F)$. It is a tensor
  network with empty boundary and a socket vertex $S_i \in V(D)$ for
  each input socket $E_i$, where $i=1,2,\ldots,\ell$. Its tensor is
  $T(D) = F(X^{(1)}, X^{(2)},\ldots, X^{(\ell)})$ where $X^{(i)} =
  T(S_i)$ for $i=1,2,\ldots,\ell$.

  By Lemma~\ref{lem:binary}, a minimum-cost execution of $D$
  can be represented by a rooted binary tree $\td$,
  where the set of leaves of $\td$ are $V(D)$ and each inner
  vertex represents the vertex obtained by contracting  its two children. Let $\ts$ be the unique socket-tree of $F$ that is obtained as a topological
  minor of $\td$. Slightly abusing the notation, we assume that the leaves of $\ts$ are the socket vertices  $S_1,S_2,\ldots,S_\ell$ instead of the sockets $E_1,E_2,\ldots,E_\ell$.
  To establish the lemma, it suffices to show that
  $\td$ has cost at least $w(\ts)$, since $w(\ts) \geq w(F)$.

  Let $e = \{x_R, x_C\} \in E(\ts)$ be an edge of the socket tree
  $\ts$ with $\rk(M(\ts, e)) = w(\ts)$, and let $\widetilde{e}$ be an edge of
  the execution tree $\td$ in the subdivision of $e$ appearing in
  $\td$. Without loss of generality we may assume that $\widetilde{e}$ is directed from the part of $\td$ corresponding to $x_R$ towards the part corresponding to $x_C$ (if not, simply switch names of $x_R$ and $x_C$).  Define  $S_R = L(x_R, e)$ and
  $S_C = L(x_C, e)$.  Let $W_R \subseteq V(D)$ be the set of
  non-socket vertices of $D$ that appear on the same side of $\widetilde{e}$ in $\td$ with socket vertices $S_R$ and let
  $W_C$ be the set of remaining non-socket vertices of $D$.  See Figure~\ref{fig:lower bound proof} for an illustration of all these definitions.
  Finally, let $D' = D / S_R / S_C / W_R / W_C$ be the
  result of contracting each of these four sets of vertices of
  $D$.
  For notational convenience, we identify the four vertices of
  the new network with the four subsets $S_R, S_C, W_R, W_C$.

  \begin{figure}
    \begin{subfigure}[b]{0.52\textwidth}
      \centering
      \includegraphics[scale=0.86]{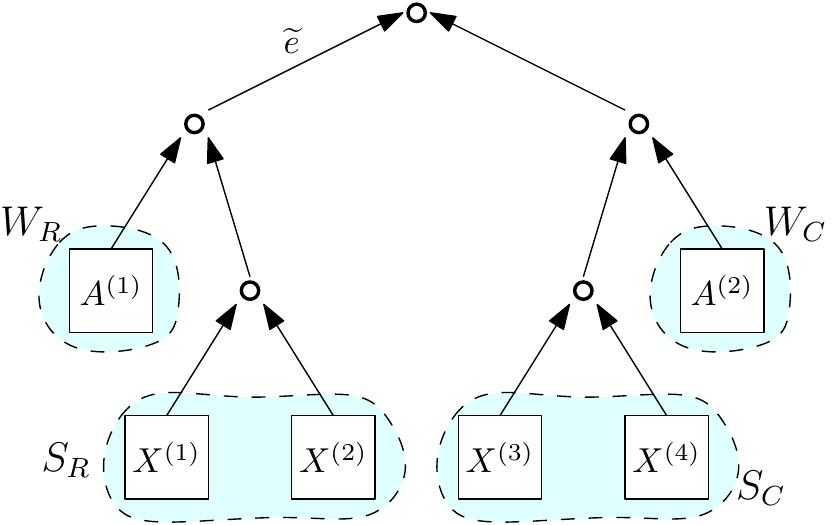}
      \caption{Example of a possible execution tree $\td$.  Given the
        choice of $e$ in the corresponding socket tree $\ts$ shown on
        the right there are four possible choices of $\widetilde{e}$.}
    \end{subfigure}
    \qquad
    \begin{subfigure}[b]{0.4\textwidth}
      \centering
      \includegraphics{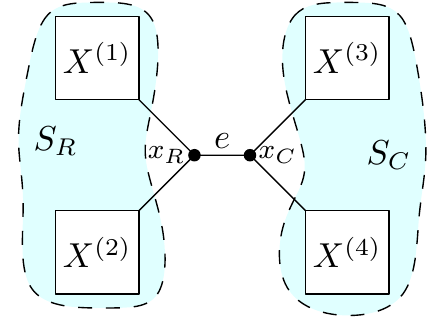}
      \caption{The corresponding socket tree $\ts$. The exact choice of $\widetilde{e}$ in $\td$ determines which part of the
        cut is the $x_R$ part, and which is the $x_C$ part.}
      \label{fig:lb proof b}
    \end{subfigure}
    \caption{Illustration of the notation used for the execution and socket trees.
    }
    \label{fig:lower bound proof}
  \end{figure}

  Now, the tensor $P = T(D'[W_R\cup S_R])$ appears as an intermediate result in
  the execution $\td$,\footnote{Note that the same is not true for the tensor $T(D'[W_C \cup S_C])$.} hence the volume of $P$ is a lower bound on the cost of $\td$.

  We group the modes of $D'$ incident on $S_R$ or $W_R$ as shown in
  Figure~\ref{fig:lower bound proof 2}: $E_{SW}$ are all modes in $D'$
  incident exactly upon $S_R$ and $W_R$, $E_{WC}$ are all modes
  incident on $W_R$ but not on $S_R$, $E_{SC}$ are all modes incident
  on $S_R$ but not $W_R$, and finally $E_{SWC}$ are all modes incident
  upon $S_R$, $W_R$, and at least one of $S_C$ or $W_C$.  Write
  $E_{S} = E_{SW} \cup E_{SC} \cup E_{SWC}$ for the modes incident on
  $S_R$, and similarly $E_{C} = E_{WC} \cup E_{SC} \cup E_{SWC}$ for
  all modes incident upon at least one of $S_R$/$W_R$ and at least one
  of $S_C$/$W_C$.  Note that $|J(E_C)|$ is precisely the volume of $P$
  which we aim to lower bound.

  \begin{figure}
    \centering
    \includegraphics[scale=1.1]{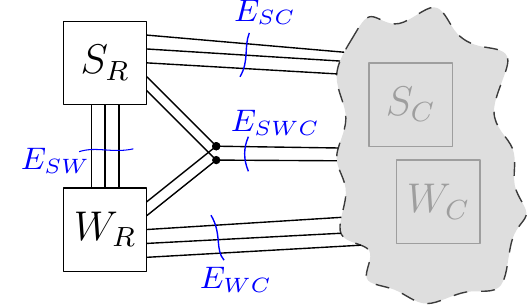}
    \caption{Illustration of $D'$.  We group the modes of $D'$ based on how they connect $S_R$, $S_C$, and the ``$C$ part'' of $D'$.}
    \label{fig:lower bound proof 2}
  \end{figure}

Define a matrix $A \in \F^{J(E_S)} \times \F^{J(E_{C})}$ as follows.  We identify its row indices $i \in J(E_S)$ as being triples $i = (i_{SW},i_{SC},i_{SWC}) \in J(E_{SW}) \times J(E_{SC}) \times J(E_{SWC})$ and similarly its column indices $j \in J(E_{C})$ are triples $j = (j_{SC},j_{WC},j_{SWC}) \in J(E_{SC}) \times J(E_{WC}) \times J(E_{SWC})$.  Then the entries of $A$ are
  \begin{equation*}
    A_{(i_{SW},i_{SC},i_{SWC}),(j_{SC},j_{WC},j_{SWC})} =
  \begin{cases}
    T(D'[W_R])_{i_{SW},j_{WC},j_{SWC}} & \!\!\text{if } i_{SC} = j_{SC} \wedge i_{SWC}=j_{SWC},\\
    0 & \!\!\text{otherwise},
  \end{cases}
  \end{equation*}
  In the case when $E_S = E_{SW}$ (i.e., all modes incident on $S_R$ connect only to $W_R$), $A$ is simply a flattening of $T(D'[W_R])$. Recall that $T(D'[S_R]) \in \prod_{e \in E_S}\F^{J(e)}$. Then for every $j = (j_{SC},j_{WC},j_{SWC}) \in J(E_C)$, we have
  \begin{align*}
    \sum_{i \in J(E_S)} A_{i,j} T(D'[S_R])_{i} &= \sum_{i_{SW} \in J(E_{SW})} A_{(i_{SW},j_{SC},j_{SWC}),j} T(D'[S_R])_{i_{SW},j_{SC},j_{SWC}} \\
    &= \sum_{i_{SW}} T(D'[W_R])_{i_{SW},j_{WC},j_{SWC}} T(D'[S_R])_{i_{SW},j_{SC},j_{SWC}}\\
    &= P_{j_{SC},j_{WC},j_{SWC}} = P_j
  \end{align*}
  (recall that $P$ is the contraction of $T(D'[W_R])$ and $T(D'[S_R])$).
  Viewing $T(D'[S_R])$ as a row vector in $\F^{J(E_S)}$ we see that
  $P$ is simply the
  vector-matrix product $P = T(D'[S_R]) \cdot A \in \F^{J(E_C)}$.

  Symmetrically, for the other half of $D'$, we can write $Q =
  T(D'[W_C\cup S_C])$ as a matrix-vector product $Q = B \cdot T(D'[S_C]) \in \F^{J(E_{C})}$ where $B$ is a matrix corresponding to $T(D'[W_S])$ analogously to how $A$ corresponds to $T(D'[W_R])$.

  Thus we have $T(D) = T(D'[S_R]) \cdot A \cdot
  B \cdot T(D'[S_C])$.  Recall that for each
  socket vertex $S_i$ in the original network $D$, we have $T(S_i) = X^{(i)}$.  Denoting $X_R = T(D'[S_R])$ and $X_C = T(D'[S_C])$, we get $X_R = \bigotimes_{S_i \in S_R} X^{(i)}$ and $X_C = \bigotimes_{S_i \in S_C} X^{(i)}$.\footnote{These identities use the fact that $D$ is derived from a \emph{non-degenerate} network $D^*$ for $\hat{T}(F)$.  In particular, every mode in the network $D$ is incident upon at least one non-socket vertex, hence all modes incident upon $S_R$ are boundary modes in $D'[S_R]$, and similarly for $S_C$.}
    Hence
  \[
  F(X^{(1)}, X^{(2)}, \ldots, X^{(\ell)}) = X_R \cdot A \cdot B \cdot X_C.
  \]
  It follows that $A \cdot B$ is the flattening of $T(F)$
  to a matrix with rows indexed by the sockets in $S_R$ and columns
  indexed by the sockets in $S_C$.  But this flattening is precisely
  the matrix $M(\ts, e)$, implying that $|J(E_{C)}| \ge \rk(M(\ts, e)) =
  w(\ts)$, as desired.
\end{Proof}

\section{Lower bounds for socket-width}
\label{sec:socket-width lower bound}

In this section we establish lower bounds on socket-width for concrete
maps.

\subsection{Determinant and permanent}
\label{sec:det and per socket width lower bounds}
\label{sec:permanent barrier}

We now prove lower bounds for the socket width of the determinant and
permanent.
Let us start with the following trivial observation.

\begin{Claim}
  \label{claim:balanced edge}
  For any socket tree $\ts$ with $n \ge 2$ leaves, there is an edge $e = \{x_R, x_C\}$
  such that $n/3 \le |L(x_R, e)| < 2n/3$.
\end{Claim}

\begin{Proof}
  Consider all edges $e = \{x_R, x_C\}$ such that $|L(x_R, e)| \ge n/3$.  At
  least one such edge certainly exists since $n \ge 2$. Indeed, an edge
  $e = \{x_R, x_C\}$ incident to a leaf $x_C$ has $|L(x_R, e)| = n-1$
  leaves.  Among these, choose an edge such that $\ts(x_R, e)$ is of minimal size.
  Assume for contradiction that $|L(x_R, e)| \ge 2n/3$.  But then one of
  the two subtrees of $x_R$ in $\ts(x_R, e)$ must have at least $n/3$
  leaves, and since they are smaller than $\ts(x_R, e)$, this contradicts
  the minimality of $\ts(x_R, e)$.
\end{Proof}

\begin{Lem}
For every positive integer $n\geq 2$, the socket-width of both the determinant and the permanent of an $n\times n$ matrix is at least $\binom{n}{\lceil n/3 \rceil}$.
\end{Lem}
\begin{Proof}
  Let $\ts$ be a socket tree for the permanent, let $e = \{x_R,
  x_C\}$ be an arbitrary edge in $\ts$, and let $k = |L(x_R, e)|$ be the number of
  leaves in $\ts(x_R, e)$.  We now show that $\rk(M(\ts, e)) \ge
  \binom{n}{k}$.

  Recall that the sockets of $\ts$ are the $n$ rows of the input
  matrix.  Without loss of generality, number the rows so that the
  leaves of $\ts(x_R, e)$ are rows $1$ to $k$, and the leaves of $\ts(x_C,
  e)$ are rows $k+1$ to $n$.  The rows of $M(\ts, e)$ are indexed by
  a tuple $j_R = (i_1, \ldots, i_k) \in [n]^k$, and the columns are indexed by a
  tuple $j_C = (i_{k+1}, \ldots, i_n) \in [n]^{n-k}$.  The entry of
  $M(\ts, e)$ at position $(j_R, j_C)$ is $1$ if $(i_1, i_2, \ldots,
  i_n)$ is a permutation.  For any set $U \in \binom{[n]}{k}$, let $u_1 <
  u_2 < \ldots < u_k$ be the elements of $U$ in ascending order and
  define $j_R(U) = (u_1, \ldots, u_k)$; let $u_{k+1} < u_{k+2} < \ldots < u_{n}$ be the elements of $[n]
  \setminus U$ in ascending order and define $j_C(U) = (u_{k+1}, u_{k+2}, \ldots,
  u_n)$.

  For $U_1, U_2 \in \binom{[n]}{k}$, the entry of $M(\ts, e)$ at position $(j_R(U_1),
  j_C(U_2))$ equals $1$ if $U_1 = U_2$, and $0$ otherwise. This
  induces a $\binom{n}{k} \times \binom{n}{k}$ identity submatrix
  of $M(\ts, e)$, implying $\rk(M(\ts, e)) \ge \binom{n}{k}$.

  By Claim~\ref{claim:balanced edge}, $\ts$ has an edge $e$ with $n/3
  \le k < 2n/3$.  For that edge, $\rk(M(\ts, e)) \ge \binom{n}{k} \ge \binom{n}{\lceil n/3 \rceil}$, which completes the proof for the permanent.

  For the determinant, the only change is that some
  entries of $M(\ts, e)$ become $-1$ instead of $1$, but this only changes the
  identified submatrix from an identity matrix to a diagonal matrix and
  in particular does not change its rank.
\end{Proof}

The preceding proof is similar to a lower bound by Nisan
(\cite{Nisan91}, Lemma 2) used to obtain lower bounds for algebraic
branching programs.  But the lower bounds obtained there can be made
as sharp as $\Omega(2^n)$ whereas in our setting, we cannot rule out
the possibility of a tensor network that avoids splitting the $n$
variables in two approximately equal size parts.  This means that the best we can
obtain with our current method is $\binom{n}{\lceil n/3 \rceil}$
instead of $\binom{n}{n/2}$.

\subsection{$P$-linear forms}

Suppose $\ts$ is a socket tree for a $P$-linear form for a $k$-uniform
hypergraph $P$ on $v$ vertices.
Recall that the sockets of this form correspond to the elements
of $E(P) \subseteq \binom{[v]}{k}$.
Given an edge $e \in E(\ts)$ and a vertex $x
\in V(\ts)$, we write \[U(x, e)=\bigcup_{S \in L(x,e)} S \subseteq [v].\]

\begin{Claim}
  \label{claim:binom form sockwidth lb 1}
  Let $\ts$ be a socket tree for the $P$-linear form of order $n$.
  Let $e =
  \{x_R, x_C\} \in E(\ts)$ and suppose $|U(x_C, e) \cap U(x_R, e)|= u$.
  Then the socket width of $\ts$ is at least $w(\ts) \ge n^{u}$.

\end{Claim}

\begin{proof}
  Let $m = n^u$. We show that $\rk(M(\ts, e)) \ge
  m$ by identifying an $m \times m$ identity submatrix of $M(\ts, e)$.

  Define $E_R = \{\,(S, i)\,|\,S \in L(x_R, e), i \in S\,\}$
  to be the modes contained in the sockets on the $x_R$ side of $e$, and analogously $E_C
  = \{\,(S, i)\,|\,S \in L(x_C, e), i \in S\,\}$.

  The rows of $M(\ts, e)$ are indexed by $J(E_R)$ and the columns are indexed by $J(E_C)$. We will consider an $m \times m$ submatrix of $M(\ts, e)$ whose rows and columns are indexed by the elements $\sigma \in \prod_{i \in A} [n_i]$. More specifically, each row of the submatrix is indexed by $j_R$ where $j_R|_{(S,i)}=\sigma|_i$ for $i \in A$, and $j_R|_{(S,i)} = 1$ for $i \not \in A$. Each column is indexed by $j_C$ where $j_C|_{(S,i)}=\sigma|_i$ for $i \in A$ and $j_C|_{(S,i)}=1$  for $i \not \in A$.

The value of the submatrix at position $(j_{R}, j_{C})$ is $1$ if there exists $\sigma \in \prod_{i \in A} [n_i]$ such that $j_R|_{(S,i)}=\sigma|_i$ and $j_C|_{(S,i)}=\sigma|_i$ for all $i \in A$ and $0$ otherwise. We obtain
  an $m \times m$ identity matrix as desired.
\end{proof}

From this claim, the branchwidth-based lower bound is immediate.

\begin{Lem}
  \label{lem:lb-binom-2}
  For any hypergraph $P$, the socket width of the $P$-linear form
  of order $n$ is at least $n^{\bw(P)}$.
\end{Lem}
\begin{Proof}
  Any socket tree for a $P$-linear form can be directly
  viewed as a branch decomposition of $P$.  Thus, by definition every
  socket tree for the form has an edge $e = \{x_R, x_C\}$ where
  $|U(x_C,e) \cap U(x_R,e)| \ge \bw(P)$, and the lemma now follows
  from \cref{claim:binom form sockwidth lb 1}.
\end{Proof}

For counting homomorphisms of hypergraph cliques, that is,
the $P$-linear form for the complete $k$-uniform hypergraph
$P=\binom{[v]}{k}$ on $v$ vertices, we need the
following simple lower bound on the branchwidth of complete
hypergraphs (which is most likely known, but we are not aware of a
reference).

\begin{Lem}
  For $v > k \ge 3$, the complete $k$-uniform hypergraph on $v$ vertices has branchwidth $v$.
\end{Lem}

\begin{Proof}
  Let $T$ be an arbitrary branch decomposition of the hypergraph.

  First, we note that for every edge $e = \{x_R, x_C\}$, either
  $U(x_R, e) = [v]$ or $U(x_C, e) = [v]$.  Indeed, suppose for contradiction
  that there is an edge $e = \{x_R, x_C\}$ with
  $i_1 \not\in U(x_R, e)$ and $i_2 \not\in U(x_C, e)$ for some not
  necessarily distinct $i_1, i_2 \in [v]$.  Since $k \ge 2$, there exists
  an edge $S \supseteq \{i_1, i_2\}$, and that edge must appear in
  either $\ts(x_R, e)$  or $\ts(x_C,  e)$.

  Now, if $U(x_R, e) = U(x_C, e) = [v]$ for some $e = \{x_R, x_C\}$,
  then we are done, so suppose for contradiction that this is not the case.
  Direct each edge towards the subtree which covers $[v]$.  For
  example, if $U(x_R, e) = [v]$ and $U(x_C, e) \ne [v]$ then $e$ is
  directed towards $x_R$.

  Since a tree is acyclic, there must be some vertex $x^*$ such that
  all edges incident to $x^*$ are directed towards $x^*$.  The vertex
  $x^*$ cannot be a leaf, because the subtree consisting of a leaf
  contains a single edge, which covers $k < v$ elements.  Thus $x^*$
  has degree $3$.  Let the three edges incident to $x^*$ be $e_1 =
  \{x^*, y_1\}$, $e_2 = \{x^*, y_2\}$ and $e_3 = \{x^*, y_3\}$.  Since
  all the edges are directed towards $x^*$ there are $i_1, i_2, i_3
  \in [v]$ such that $i_1 \not\in S(y_1, e_1)$, $i_2 \not\in S(y_2,
  e_2)$ and $i_3 \not\in S(y_3, e_3)$.  Since $k \ge 3$ there exists a
  edge $S \supseteq \{i_1,i_2,i_3\}$.  But now the edge $S$ cannot appear in any of the subtrees rooted at $y_1$, $y_2$, or $y_3$.
  Since these three subtrees together cover all leaves, this yields
  the desired contradiction and we conclude that there must exist an
  edge $e$ such that $U(x_R, e) = U(x_C, e) = [v]$ and therefore the branchwidth is $v$.
\end{Proof}

\subsection{Kruskal operator}

We say that an $\ell$-linear Kruskal operator is {\em $n$-uniform} if
the lengths of the modes satisfy $n=n_1=n_2=\cdots=n_\ell$.

\begin{Lem}
For positive integers $n$ and $r$, the socket-width of an $n$-uniform $\ell$-linear Kruskal operator is at least $\max\left(n^{\ell}, n^{\lceil\ell/2\rceil}r\right)$.
\end{Lem}
\begin{Proof}
  Let $\ts$ be an arbitrary socket tree for the
  operator.  One of the leaves is the output socket representing the
  tensor $B$ of shape $n \times n \times \cdots \times n$ ($\ell$
  times) that is obtained by applying the Kruskal operator to the input matrices
  $A^{(i)}$ of shape $n \times r$, where $1 \le i \le \ell$.

  Consider the neighbor $x$ of the output socket.  It has three
  neighbors: the output socket leaf, and two other vertices $x_R$ and
  $x_C$.  Either the subtree rooted at $x_R$ or the one rooted at $x_C$
  must contain $u \geq \lceil \ell/2 \rceil$ input sockets (since
  together they contain all $\ell$ input sockets).  Let $e$ be the
  edge leading to that subtree.

  We claim that $\rk(M(\ts, e)) \ge n^u r$.  Suppose without loss of
  generality that the input sockets $A^{(1)}, \ldots, A^{(u)}$ are in
  the subtree rooted at $x_R$, and that the input sockets $A^{(u+1)},
  \ldots, A^{(\ell)}$ are in the subtree rooted at $x$ together with
  the output socket.

  Each row of $M(\ts, e)$ is indexed by sequences of $2u$ indices $(i_1, j_1, \ldots, i_u, j_u)$ where each
  $i \in [n]$ and each $j \in [r]$. Each column is indexed by
  a sequence of $\ell + 2(n-u)$ indices $(i_1', \ldots, i_\ell';$ $
  i_{u+1}, j_{u+1}, \ldots, i_{\ell}, j_{\ell})$ where each $i, i' \in [n]$
  and each $j \in [r]$.  An entry of $M(\ts, e)$ is $1$ if
  and only if $i_1' = i_1, i_2' = i_2, \ldots, i_\ell' = i_\ell$, and
  $j_1 = j_2 = \ldots = j_\ell$.

  For any $i_1, \ldots, i_u \in [n]$ and $j \in [r]$, consider the row index $(i_1, j, i_2, j, \ldots, i_u, j)$ and the column
  index $(i_1, i_2, \ldots, i_u, 1, 1, \ldots, 1; 1, j, 1, j, \ldots,
  1, j)$.  The
  $n^u r \times n^ur$ submatrix of $M(\ts, e)$ induced by these sets of
  row and column indices is the identity matrix, thus $\rk(M(\ts, e))
  \ge n^u r$ as desired.

  For the $n^{\ell}$ lower bound, we
  instead consider the edge $e'$ joining $x$ with the output socket.
  Now the rows are indexed by $(i_1, j_1, i_2, j_2, \ldots, i_\ell, j_\ell)$  and the
  columns by $(i_1', \ldots, i_\ell')$.  The
  $n^\ell \times n^\ell$ identity submatrix is obtained by taking for
  every $i_1, \ldots, i_\ell \in [n]$, the row $(i_1, 1, i_2, 1, \ldots, i_\ell, 1)$ and the column $(i_1, \ldots, i_\ell)$.
\end{Proof}

%%%%%%%%%%%%%%%%%%%%%%%%%%%%%%%%%%%%%%%%%%%%%%%%%%%%%%%%%%%%%%%%%% Appendix %%%

\appendix

\section{Minimum-cost execution}

\label{sect:exec}

This appendix summarizes some results on tensor network contractions and minimum-cost executions. These results are not original work, but for the sake of completeness we present them here.

\subsection{Invariance property}

In the following lemma, we will show that the tensor of a network is equal to the tensor of any network that is obtained by a contraction from the original network. In particular, this implies that any execution gives the same tensor.

\begin{Lem}[Invariance]
\label{lem:invariance}
Let $D$ be a tensor network. For all nonempty $W\subseteq V(D)$ it holds that $T(D)=T(D/W)$.
\end{Lem}
\begin{Proof}
Let $W\subseteq V(D)$ be nonempty and let $i\in J(B(D/W))=J(B(D))$.
From \eqref{eq:induced}, \eqref{eq:tensor-rep}, and \eqref{eq:contraction},
it follows that
\[
\begin{split}
T(D/W)_i
&=
\sum_{j\in J(E(D/W)\setminus B(D/W))}
\prod_{v\in (V(D)\setminus W) \cup \{w\}}T(v)_{ij}\\
&=
\sum_{j\in J(E(D/W)\setminus B(D/W))}
T(w)_{ij}\prod_{v\in V(D)\setminus W}T(v)_{ij}\\
&=
\sum_{j\in J(E(D/W)\setminus B(D))}
T(D[W])_{ij}\prod_{v\in V(D)\setminus W}T(v)_{ij}\\
&=
\sum_{j\in J(E(D/W)\setminus B(D))}\
\sum_{j'\in J(E(D[W])\setminus B(D[W]))}\
\prod_{w\in W}T(w)_{ijj'}
\prod_{v\in V(D)\setminus W}T(v)_{ij}\\
&=
\sum_{j\in J(E(D/W)\setminus B(D))}\
\sum_{j'\in J(E(D[W])\setminus B(D[W]))}\
\prod_{v\in V(D)}T(v)_{ijj'}\\
&=
\sum_{j''\in J(E(D)\setminus B(D))}\
\prod_{v\in V(D)}T(v)_{ij''}\\
&=T(D)_i\,.
\end{split}
\]
\end{Proof}

\subsection{The structure of a minimum-cost execution}

In this section, we analyze the structure of a minimum-cost execution. In particular, we will prove Lemma~\ref{lem:binary} that states that each contracted set has size at most two and also show that one can always contract adjacent vertices in a network.

\BinaryExecutionLemma*

\begin{Proof}
If $D$ contains loops, we may assume that
a minimum-cost execution first removes all the loops by contracting singleton
vertices incident to loops. Indeed, the cost of contracting
a singleton vertex is the volume of the tensor associated to it. Since the result of an execution is a single tensor, then every vertex has to be contained in a contracted set of an execution and none of the hyperedges incident to a vertex cannot be removed before the vertex is contained in a contracted set. Hence, the volume of any tensor in the tensor network is a lower
bound for the cost of any execution of $D$ and we may contract singleton vertices.

So let us assume that $D$ is loopless.
Suppose that a minimum-cost execution of $D$ contains a contraction by
a set $W=\{w_1,w_2,\ldots,w_s\}$ of size at least $s\geq 3$, and let $w$ be the new vertex after this contraction. Then we can replace the contraction by $W$ with
two contractions by $W'=\{w_1,w_2,\ldots,w_{s-1}\}$ and $W''=\{w,w_s\}$ without increasing the
cost of the execution. The cost of contracting $W'$ is less or equal than the cost of contracting $W$, because every hyperedge incident to $W'$ is also incident to $W$. The cost of contracting $W''$ is less or equal than the cost of contracting $W$, because the set of hyperedges incident to $w$ is contained in the set of hyperedges incident to $W'$. We repeat this procedure until all contracted sets in the
execution have size at most two.
\end{Proof}

Two tensors in a tensor network are called {\em adjacent}
if they are incident to a common mode.

\begin{Lem}[Execution by contracting adjacent tensors]
\label{lem:adjacent}
Let $D$ be a loopless tensor network that is connected as a hypergraph.
Then, there exists a minimum-cost execution of $D$ such that
each contracted set has size two and consists of adjacent vertices.
\end{Lem}
\begin{Proof}
Consider a minimum-cost execution of $D$ such that each contracted
set has size two; such an execution exists by Lemma~\ref{lem:binary}.
If all contracted sets in the execution consist of adjacent vertices,
we are done. Otherwise, consider a contraction of two vertices, $u$ and $v$,
such that $u$ and $v$ are not adjacent when they are
contracted to yield the vertex $uv$.
Consider the steps of the execution after this contraction step.
Let us call such a step a {\em relevant} step
if it involves a descendant of the vertex $uv$. Let us modify the
execution as follows. First, delete the contraction of $u$ and $v$ from
the execution. Then, for each relevant step in execution order, replace
the descendant of the vertex $uv$ with the current descendant of
either $u$ or $v$ so that the contraction becomes a contraction of
adjacent vertices whenever possible. If the descendants of $u$ and $v$
become adjacent after a relevant step, contract the descendants
of $u$ and $v$, and then continue the execution without further changes.
Since $D$ is connected, the descendants of $u$ and $v$
must eventually become adjacent.

We will show that this modification of the execution gives again an execution and that it has cost no larger than the original minimum-cost
execution.
In the modification of the execution, let the contraction sets containing $u$ or a descendant of $u$ and not containing a descendant of $v$ be $\{u,w_1\},\{uw_1,w_2\}$, $\ldots$, $\{uw_1\ldots w_{s-1},w_s\}$; similarly, let the contraction sets containing $v$ or a descendant of $v$ and not containing a descendant of $u$ be $\{v,z_1\}$, $\{vz_1,z_2\},\ldots$, $\{vz_1\ldots z_{t-1},z_t\}$. Here $w_1,\ldots,w_s,z_1,\ldots,z_t$ can be vertices of $D$ or vertices obtained after contraction steps.  Contracting $\{uw_1\ldots w_s,vz_1 \ldots z_t\}$  gives the vertex $uw_1 \ldots w_s v z_1 \ldots z_t$, which  appears also in the original execution. Hence, after certain steps the tensor networks in the original execution and in the modification are the same (after making necessary non-relevant contractions) and the modification of the original execution is also an execution.

First we consider the cost of contracting $\{uw_1\ldots w_{i-1},w_i\}$ in the modified execution where $1 \leq i \leq s$. There exists a $j$ satisfying $1 \leq j \leq t$ such that in the original execution we contract $\{uw_1 \ldots w_{i-1}vz_1 \ldots z_j,w_i\}$. Then the cost of contracting $\{uw_1\ldots w_{i-1},w_i\}$ is less or equal that the cost of contracting $\{uw_1 \ldots w_{i-1}vz_1 \ldots z_j,w_i\}$, because every hyperedge incident to the vertex $uw_1\ldots w_{i-1}$ in the modified execution has to be incident to $uw_1 \ldots w_{i-1}vz_1 \ldots z_j$ in the original execution. Otherwise, there would be an hyperedge in $D$ that is incident only to vertices in $\{u,w_1, \ldots, w_{i-1},v,z_1, \ldots, z_j\}$, and in particular both to vertices in $\{u,w_1, \ldots, w_{i-1}\}$ and in $\{v,z_1, \ldots, z_j\}$. This is impossible, because then $uw_1\ldots w_{i-1}$ and $vz_1 \ldots z_j$ would be adjacent vertices in the modified execution, but by assumption $uw_1\ldots w_s$ and $vz_1 \ldots z_t$ are the first adjacent descendants of $u$ and $v$ in the modified execution. Similarly, we can show that the cost of contracting $\{vz_1\ldots z_{j-1},z_j\}$ in the modified execution for $1 \leq j \leq t$ is less than the cost of contracting a set in the original execution.

Second we consider the cost of contracting $\{uw_1\ldots w_s,vz_1 \ldots z_t\}$ in the modified execution. Without loss of generality, we can assume that in the original execution the vertex $uw_1\ldots w_svz_1 \ldots z_t$ is obtained from contracting $\{uw_1\ldots w_svz_1 \ldots z_{t-1}, z_t\}$. We will show that the cost of contracting $\{uw_1\ldots w_s,vz_1 \ldots z_t\}$ in the modified execution is less or equal than the cost of contracting $\{uw_1\ldots w_svz_1 \ldots z_{t-1}, z_t\}$ in the original execution. Indeed, every hyperedge incident to $uw_1\ldots w_s$ in the modified execution is incident to $uw_1\ldots w_svz_1 \ldots z_{t-1}$ in the original execution, because otherwise there would be an hyperedge in $D$ that is incident only to vertices in $\{u,w_1, \ldots, w_s,v,z_1, \ldots, z_{t-1}\}$, and in particular both to vertices in $\{u,w_1, \ldots, w_s\}$ and $\{v,z_1, \ldots, z_{t-1}\}$. However, by assumption $uw_1\ldots w_s$ and $vz_1 \ldots z_t$ are the first adjacent descendants of $u$ and $v$ in the modified execution. Similarly, every hyperedge incident to $vz_1 \ldots z_t$ in the modified execution is incident to $uw_1\ldots w_svz_1 \ldots z_{t-1}$ or $z_t$ in the original execution, because otherwise there would be an hyperedge in $D$ that is incident only to vertices in $\{u,w_1, \ldots, w_s,v,z_1, \ldots, z_{t-1}\}$, and in particular both to vertices in $\{u,w_1, \ldots, w_s\}$ and $\{v,z_1, \ldots, z_{t-1}\}$. This contradicts that $uw_1\ldots w_s$ and $vz_1 \ldots z_t$ are the first adjacent descendants of $u$ and $v$ in the modified execution.

The modified execution consists of at least one less contraction of nonadjacent
vertices. Repeating this procedure until there are no contractions of nonadjacent vertices completes the lemma.
\end{Proof}

\subsection{Finding a minimum-cost execution}

In this section, we will show a recurrence relation for the cost of a tensor network. To this end, let $D$ be a loopless connected tensor network $D$ with at least two tensors.
Let us write $\mathscr{C}(D)$ for the set of all $W\subseteq V(D)$ such
that $D[W]$ is connected.

\begin{Lem}
Let $D$ be a loopless connected tensor network $D$ with at least two tensors. The cost of $D$ satisfies the recurrence
\begin{equation}
\label{eq:minimum-cost}
c(D)
=\!\!\!\!
\min_{\substack{W\in\mathscr{C}(D)\\ V(D)\setminus W\in\mathscr{C}(D)}}
\!\!\!\!
\max\bigl(c(D[W]),c(D[V(D)\setminus W]),c(D/W/(V(D)\setminus W))\bigr)\,.
\end{equation}
\end{Lem}
\begin{Proof}
We recall that the cost $c(D)$ of $D$ is the cost of minimum-cost execution of $D$. We can view a minimum-cost execution of $D$ as a rooted tree
such that the root has degree two, all non-root internal vertices have
degree three, and each leaf vertex is a tensor of $D$.
By Lemma~\ref{lem:adjacent}, the two subtrees rooted at the neighbours of the
root define two disjoint sets $U,W\subseteq V(D)$ such that $U\cup W=V(D)$
and both $D[U]$ and $D[W]$ are connected. The cost of the contraction
at the root is $c(D/W/U)$, and the subtrees rooted at $U$ and $W$
cost $c(D[U])$ and $c(D[W])$, respectively. The recurrence considers
all the partitions to connected sets, including the optimum.
\end{Proof}

%%%%%%%%%%%%%%%%%%%%%%%%%%%%%%%%%%%%%%%%%%%%%%%%%%%%%%%%%%%%%%%% References %%%

\bibliography{paper-5}

%%%%%%%%%%%%%%%%%%%%%%%%%%%%%%%%%%%%%%%%%%%%%%%%%%%%%%%%%%%% Documents ends %%%

\end{document}